\documentclass[letterpaper,11pt]{extarticle}
\usepackage[left=2.5cm,right=2.5cm,top=2.5cm,bottom=2.5cm]{geometry}

\usepackage{hyperref}
\usepackage{tocloft}

\usepackage{titlesec}
\titlespacing*{\paragraph}{0pt}{0.5ex}{0.5em}

\usepackage{amsmath,mathtools,amssymb,amsthm,xcolor,url,esvect,paralist}
\usepackage{tikz}
\usetikzlibrary{arrows,arrows.meta,shapes,calc,decorations} 
\usetikzlibrary{decorations.pathmorphing,decorations.pathreplacing,decorations.markings}

\usepackage{subcaption}

\usepackage{circledsteps}
\usepackage{bm}

\usepackage{verbatim}

\tikzset{big dots/.style={
  decoration={markings,
    mark=between positions 0.1 and 0.9 step 3pt with {
      \fill circle[radius=0.5pt];
    }
  },
  postaction={decorate}
}}

\swapnumbers
\numberwithin{equation}{section}

\newtheorem{theorem/}[equation]{Theorem}
\newenvironment{theorem}
  {%
   \pushQED{\qed}\begin{theorem/}}
  {\popQED\end{theorem/}}

\newtheorem{lemma/}[equation]{Lemma}

\newtheorem{claim/}[equation]{Claim}

\newtheorem{corollary/}[equation]{Corollary}
\newenvironment{corollary}
  {%
   \pushQED{\qed}\begin{corollary/}}
  {\popQED\end{corollary/}}

\newtheorem{proposition/}[equation]{Proposition}
\newenvironment{proposition}
  {%
   \pushQED{\qed}\begin{proposition/}}
  {\popQED\end{proposition/}}

\theoremstyle{definition}
\newtheorem{definition/}[equation]{Definition}

\newtheorem{example/}[equation]{Example}

\newtheorem{remark/}[equation]{Remark}

\newtheorem*{example*}{Example}
\newtheorem{caveat/}[equation]{Caveat}

\newtheorem*{caveat*}{Caveat}
\newtheorem{conjecture/}[equation]{Conjecture}

\newtheorem*{conjecture*}{Conjecture}
\newtheorem{openproblem/}[equation]{Open Problem}
\newenvironment{openproblem}
  {%
   \pushQED{\qed}\begin{openproblem/}}
  {\popQED\end{openproblem/}}
\newtheorem*{openproblem*}{Open Problem}

\newcommand{\poly}{\textup{\textsf{poly}}}

\newcommand{\IZ}{\mathbb{Z}}

\newcommand{\la}{\lambda}
\newcommand{\IN}{\mathbb{N}}
\newcommand{\aS}{\mathfrak{S}}

\newcommand{\GL}{\textup{GL}}
\newcommand{\diag}{\textup{diag}}
\newcommand{\Clique}{\mathrm{Clique}}
\renewcommand{\det}{\mathrm{det}}

\newcommand{\IMM}{\textup{IMM}}

\newcommand{\transpose}{\textup{\textsf{T}}}
\newcommand{\w}{\ensuremath{w}}
\newcommand{\wh}{\ensuremath{\w_{\textup{h}}}}
\newcommand{\asize}{\textup{\textsf{ac}}} 
\newcommand{\dc}{\textup{\textsf{dc}}}

\DeclareMathOperator{\tr}{tr}
\DeclareMathOperator{\sgn}{sgn}

\DeclareMathOperator{\cyco}{CC}
\DeclareMathOperator{\CCP}{CCP}
\newcommand{\cpc}{\chi}

\newcommand{\trp}{\mathrm{tr\textup{-}pow}}
\newcommand{\pow}{\mathrm{pow}}

\newcommand{\GapL}{\mathbf{\mathrm{GapL}}}

\newcommand{\VW}{\mathrm{VW}}

\newcommand{\VFPT}{\mathrm{VFPT}}
\newcommand{\VBFPT}{\mathrm{VBFPT}}
\newcommand{\VBP}{\mathrm{VBP}}
\newcommand{\VNP}{\mathrm{VNP}}

\usepackage[utf8]{inputenc}
\DeclareUnicodeCharacter{27E8}{\ensuremath{\langle}}
\DeclareUnicodeCharacter{27E9}{\ensuremath{\rangle}}

\makeatletter
\def\pgfdecoratedcontourdistance{0pt}
\pgfset{
  decoration/contour distance/.code=%
    \pgfmathsetlengthmacro\pgfdecoratedcontourdistance{#1}}
\pgfdeclaredecoration{contour lineto closed}{start}{%
  \state{start}[
    next state=draw,
    width=0pt,
    persistent precomputation=\let\pgf@decorate@firstsegmentangle\pgfdecoratedangle]{%
    \pgfpathmoveto{\pgfpointlineattime{.5}
      {\pgfqpoint{0pt}{\pgfdecoratedcontourdistance}}
      {\pgfqpoint{\pgfdecoratedinputsegmentlength}{\pgfdecoratedcontourdistance}}}%
  }%
  \state{draw}[next state=draw, width=\pgfdecoratedinputsegmentlength]{%
    \ifpgf@decorate@is@closepath@%
      \pgfmathsetmacro\pgfdecoratedangletonextinputsegment{%
        -\pgfdecoratedangle+\pgf@decorate@firstsegmentangle}%
    \fi
    \pgfmathsetlengthmacro\pgf@decoration@contour@shorten{%
      -\pgfdecoratedcontourdistance*cot(-\pgfdecoratedangletonextinputsegment/2+90)}%
    \pgfpathlineto
      {\pgfpoint{\pgfdecoratedinputsegmentlength+\pgf@decoration@contour@shorten}
      {\pgfdecoratedcontourdistance}}%
    \ifpgf@decorate@is@closepath@%
      \pgfpathclose
    \fi
  }%
  \state{final}{}%
}
\makeatother
\tikzset{
  contour/.style={
    decoration={
      name=contour lineto closed,
      contour distance=#1
    },
    decorate}}

\title{On the gradient of the coefficient of the characteristic polynomial}

\author{Christian Ikenmeyer}

\date{November 7, 2025}

\begin{document}
\raggedbottom
\maketitle

\setlength{\abovedisplayskip}{3pt}
\setlength{\belowdisplayskip}{3pt}

\begin{abstract}
We prove the bivariate Cayley--Hamilton theorem, a powerful generalization of the classical Cayley--Hamilton theorem.
The bivariate Cayley--Hamilton theorem has three direct corollaries that are usually proved independently: The classical Cayley--Hamilton theorem, the Girard-Newton identities, and the fact that the determinant and every coefficient of the characteristic polynomial has polynomially sized algebraic branching programs (ABPs) over arbitrary commutative rings.
This last fact could so far only be obtained from separate constructions, and now we get it as a direct consequence of this much more general statement.

The statement of the bivariate Cayley--Hamilton theorem involves the gradient of the coefficient of the characteristic polynomial, which is a generalization of the adjugate matrix.
Analyzing this gradient, we obtain 
another new ABP for the determinant and every coefficient of the characteristic polynomial.
This ABP has one third the size and half the width compared to the current record-holder ABP constructed by Mahajan--Vinay in 1997.
This is the first improvement on this problem for 28 years.

Our ABP is built around algebraic identities involving the first order partial derivatives of the coefficients of the characteristic polynomial, and does not use the ad-hoc combinatorial concept of clow sequences.
This answers the 26-year-old open question by Mahajan--Vinay from 1999 about the necessity of clow sequences.

We prove all results in a combinatorial way that on a first sight looks similar to Mahajan--Vinay, but it is closer to Straubing's and Zeilberger's constructions.
\end{abstract}

\bigskip
\bigskip

\noindent{\begin{tabular}{p{1.5cm}p{13.5cm}}
\footnotesize \textbf{Keywords: }
& \footnotesize
Cayley--Hamilton theorem, Girard--Newton identities, determinant, characteristic polynomial, algebraic complexity theory, GapL, parameterized complexity
\end{tabular}
}

\bigskip

\begingroup
\setlength{\cftbeforesecskip}{0.5em}
\setcounter{tocdepth}{1}
\tableofcontents
\endgroup

\begingroup
  \renewcommand\thefootnote{}
  \footnotetext{For the purpose of open access, the author has applied a Creative Commons Attribution (CC-BY) license to any Author Accepted Manuscript version arising from this submission.}%
\endgroup

\thispagestyle{empty}
\clearpage
\pagenumbering{arabic}

\section{Introduction}
\label{sec:intro}

The determinant polynomial
\[
\det_d \ = \ \sum_{\pi\in\aS_d} \sgn(\pi) \, x_{i,\pi(i)}
\]
is the most prominent of the coefficients of the characteristic polynomial.
It is ubiquitous throughout mathematics and has been studied for over 200 years.
In counting complexity theory, it is complete for the class $\GapL$,
and in algebraic complexity theory, it is complete for the class $\VBP$.
The Cayley--Hamilton theorem is one of the most fundamental results in linear algebra. It states that every square matrix is a root of its own characteristic polynomial.
If $X_n$ is the $n\times n$ matrix of variables $X_n = \Big(x_{i,j}\Big)_{\substack{1\leq i\leq n\\1 \leq j \leq n}}$,
the characteristic polynomial is defined as
\begin{equation}\label{eq:cpct}
\det(X_n+t \cdot I_d)
 \ = \ \sum_{d\geq 0} \cpc_{n,n-d} \cdot t^d
\end{equation}
where $I_n$ is the $n\times n$ identity matrix and $t$ is an indeterminate.
The Cayley--Hamilton theorem states that $\sum_{i=0}^n (-1)^i \cpc_{n,n-i} X_n^i$ is the zero matrix.
Equation~\eqref{eq:cpct} defines the homogeneous degree $d$ coefficient polynomials $\cpc_{n,d}$, which will play a major role in this paper.
From \eqref{eq:cpct} it follows that
\begin{equation}\label{eq:cpcS}
\cpc_{n,d} \ = \ \sum_{S\in\binom{\{1,\ldots,n\}}{d}}\det([X_n]_{S,S}),
\end{equation}
where $[X]_{S,S}$ is the square submatrix of $X$ with row indices $S$ and column indices $S$,
and $\binom{A}{d}$ is the set of cardinality $d$ subsets of the set~$A$.
Note that $\cpc_{d,d}=\det_d$ and $\cpc_{n,1}=\tr_n := \tr(X_n)$ is the trace.

\paragraph{Commutative rings}
The above discussion about the determinant, the characteristic polynomial, and the Cayley--Hamilton theorem works over arbitrary commutative rings~$R$, and this is the generality we work in. This is the common setup in commutative algebra, see for example \cite[\S0.1]{Eis13}: ``Nearly every ring treated in this book is commutative, and we shall generally omit the adjective''.
The specific search for algorithms that work over all commutative rings goes at least back to \cite{Val92} and the concept of Valiant's ``ARP'' (all rings polynomially-sized circuits),
and is again prominently featured in \cite{MV97}.
Let $R[\mathbf x]$ denote the polynomial ring over $R$, and let $R[\mathbf x]_d$ denote its homogeneous degree $d$ component. Note that $\cpc_{n,d}\in R[\mathbf x]_d$.

In algebraic complexity theory, constructions do not always work over arbitrary commutative rings, for example when the Girard--Newton identities are involved. This was an issue in the recent breakthrough lower bound \cite{LST25}, which was made field-independent later in \cite{For24} (best paper award at CCC).
We shed some light on the Girard--Newton identities in \S\ref{sec:corollaries}, and in fact our bivariate Cayley--Hamilton theorem is a generalization of those identities.
Field-independent replacements for field-dependent constructions are a goal in several areas of mathematics, for example algebraic geometry, representation theory, and invariant theory, see for example \cite{Eis13,AFPRW19,DV22}.

\paragraph{Motivation: Width as the bivariate perspective on ABPs}
Every homogeneous degree $d$ polynomial $f$ can be written as the bottom-right entry of a product of exactly $d$ many $n\times n$ matrices
whose entries are homogeneous linear polynomials, for some large~$n$.
The \emph{width} $\w(f)$ is defined as the smallest $n$ such that this is possible.
We call a $d$-tuple of such $n\times n$ matrices a pure algebraic branching program (pABP).
A trivial upper bound for $\w(f)$ is the number of monomials of~$f$:
In this straightforward upper bound construction, the first matrix has nonzero entries only in the last row, the last matrix has nonzero entries only in the last column, and all other matrices are diagonal.
The width is robust in the sense that it is the same notion when we study so-called homogeneous algebraic branching programs (ABPs) instead of pABPs, see Proposition~\ref{pro:wwh} below.
An ABP is a directed acyclic graph
in which each vertex has a layer from $0,\ldots,d$,
and each edge is either going from a vertex in layer $i$ to another vertex in layer $i$ and is labeled by an element of $R$ (these are called constant edges),
or the edge is going from a vertex in layer $i$ to a vertex in layer $i+1$ and is labeled by an element of $R[\mathbf x]_1$, see for example Figure~\ref{subfig:compressed}.
One vertex in layer 0 is called the source $s$, and one vertex in layer $d$ is called the sink $t$.
To an ABP $G$ we assign its computation result~$f$, which is the weighted sum over all $s$-$t$-paths,
where each summand is weighted by the product of its edge labels.
We say that $G$ computes~$f$.
The width of an ABP is defined as the maximum number of vertices in any layer $1,\ldots,d-1$.
The smallest width of an ABP computing $f$ equals $\w(f)$, see Proposition~\ref{pro:wwh}.
The ABP is a graph theoretic variant of the algebraically pure notion of the pABP, and we will see in Proposition~\ref{pro:wdetpoly} that it is sometimes convenient to work with ABPs instead of pABPs.
The width is also closely related to the so-called affine algebraic branching program size, see \S\ref{sec:appendix}.

We can think of the pair $\big(\w(f),\deg(f)\big)$ as a highly algebraically idealized way of measuring the running time requirements to evaluate~$f$, which leans heavily towards the algebraic side, but captures arithmetic circuit size up to a quasipolynomial blowup \cite{vzG87}.
The degree measures a part of the computational complexity $f$, but since for any $f\in R[\mathbf x]_d$ the degree is exactly~$d$, the width is the more interesting quantity.
However, unlike all other ways of measuring complexity, width and degree must always be viewed together as a tuple, because
there are sequences of polynomials of constant width with unbounded degree (and hence the evaluation time diverges when $d$ increases)
and
there are sequences of polynomials of constant degree with unbounded width (for example, the degree 3 $n\times n$ matrix multiplication polynomials).
There is the following duality between the degree and the width with respect to subadditivity and submultiplicativity.
For $f,h\in R[\mathbf x]_d$ we have
{
\setlength{\abovedisplayskip}{7pt}
\setlength{\belowdisplayskip}{7pt}
\begin{equation}\label{eq:duality}
\setlength{\arraycolsep}{1pt}
\begin{array}{rl@{\qquad}rl}
\deg(f+h) &\leq \max(\deg(f),\deg(h)) & \deg(fh) &\leq \deg(f)+\deg(h) \\[0.2em]
\w(f+h)   &\leq \w(f)+\w(h)          & \w(fh)   &\leq \max(\w(f),\w(h))
\end{array}
\end{equation}
}
where we used the conventions $\deg(0)=0$ and $\w(0)=0$, see \S\ref{subsec:appendixduality}.
The determinantal complexity $\dc(f)$ is defined as the smallest $n$ such that $f$ can be written as the determinant of an $n\times n$ matrix with entries from $R[\mathbf x]_0 \oplus R[\mathbf x]_1$, i.e., affine linear entries.
The fact that the entries must be affine cannot avoided here, because the determinant polynomial $\det_d$ is just single-indexed. Moreover, while $\dc(fh)\leq \dc(f)+\dc(h)$, no nice relation is known for $\dc(f+h)$ (only with a blowup, and only via a reduction to ABPs).
Fortunately, determinantal complexity is closely related to width as follows:
\begin{equation}
\label{eq:wdc}
\tfrac 1 d \, \dc(f) \ \leq \ \w(f) \ \stackrel{(\ast)}{\leq} \ d^4 \,\dc(f),
\end{equation}
where $(\ast)$ is currently only known to hold over fields, see Proposition~\ref{pro:CKV}.
Clearly, $\dc(\det_d)=d$, but $(\ast)$ cannot be used to prove $\w(\det_d)\in\poly(d)$, because the proof of $(\ast)$ uses this fact as an ingredient.
Compared to other ways of measuring computational complexity, the width can be used to study $n$ and $d$ independently, and hence it is well suited for studying bivariate complexity, which is commonly known as parameterized complexity, see \cite[Preface]{DF13}.

\paragraph{Motivation: The bivariate viewpoint in complexity theory}
The study of bivariate polynomial families where one index is the degree has recently been pioneered in algebraic complexity theory by \cite{BE19},
translating results from Boolean parameterized complexity theory \cite{DF99,Nie06,FG06,DF13} into algebraic complexity theory.
From the Boolean perspective, \cite{DF13} write ``the multivariate perspective has proved useful, even arguably essential''.
We will see in this paper that the multivariate perspective is also valuable in linear algebra.
The general theory in \cite{DF13} allows different so-called ``parameters'',
and finding the ``correct'' parameter can be an art.
But when studying homogeneous polynomials, the canonical choice for the parameter is the degree,
see \cite[Def.~4.7(2), Def.~4.15(2)]{BE19}.
Treating the degree as a separate parameter very recently made its appearance in several other papers in algebraic complexity theory \cite{DGIJL24,CKV24,vdBDGIL25}.
For double-indexed polynomials $f_{n,d}$ we define the exponents
\[
\gamma(f_{.,d}) := \limsup_{n\to\infty}\big(\log_n(\w(f_{n,d}))\big).
\]
As seen in \eqref{eq:wdc}, these exponents are invariant under exchanging the width with determinantal complexity.
They are also invariant under exchanging the width with the affine algebraic branching program size, see \eqref{eq:asizeleqw} and Proposition~\ref{pro:homogenization}.
The complexity class $\VBFPT$ consists of those $f_{n,d}$ for which the set $\{\gamma(f_{.,d})\mid d \in \IN\}$ is bounded, see also Proposition~\ref{pro:VBFPT}.
The univariate class $\VBP$ consists of \footnote{In our definition of $\VBP$, the sequences are indexed by the degree and not by the number of variables, which gives a cleaner theory, see \cite{vdBDGIL25}. This makes no difference for the fundamental complexity theoretic considerations.} those single-indexed $f_d$ for which $\w(f_d)\in\poly(d)$.
The variant of $\VBFPT$ that uses arithmetic circuit complexity instead of width is defined in \cite{BE19}, which gives rise to a potentially larger class~$\VFPT$.
One can relate different conjectures in algebraic complexity theory via the clique polynomial, for example.
Let $\Clique_{n,d}=0$ if $d\neq \binom{k+1}{2}$ for some $k$,
and let
\[
\Clique_{n,\binom{k+1}{2}} \ = \ \sum_{1 \leq v_1<v_2<\cdots<v_k\leq n} \prod_{1\leq i \leq j \leq k} x_{v_i,v_j}.
\]
Then by \cite{BE19} and \cite[Proof~of~Thm~3.10]{Bur00CaR} we have
\begin{center}
\begin{tikzpicture}[xscale=2.5,yscale=0.6]
\node at (0,0) {$\Clique_{n,d} \notin \VFPT$};
\node at (1,0) {$\Rightarrow$};
\node at (2,0) {$\Clique_{n,d} \notin \VBFPT$};
\node at (3,0) {$\Rightarrow$};
\node at (4,0) {$\Clique_{2d,d}\notin\VBP$};
\node at (0,-2) {$\VFPT\neq\VW[1]$};
\node at (1,-2) {$\Rightarrow$};
\node at (2,-2) {$\VBFPT\neq\VW[1]$};
\node at (3,-2) {$\Rightarrow$};
\node at (4,-2) {$\VBP\neq\VNP$};
\node at (0,-1) {\rotatebox{90}{$\Leftrightarrow$}};
\node at (2,-1) {\rotatebox{90}{$\Leftrightarrow$}};
\node at (4,-1) {\rotatebox{90}{$\Leftrightarrow$}};
\end{tikzpicture}
\end{center}
The definitions of $\VNP$ \cite{Bur00CvV} and $\VW[1]$ \cite{BE19} will not play a role in this paper.
The rightmost vertical equivalence holds in characteristic $\neq 2$.
In characteristic $\neq 2$, the conjecture $\VBP\neq\VNP$ is also known as Valiant's determinant vs permanent conjecture.
It can be conveniently expressed bivariately as $\w(\Clique_{n,d})\notin\poly(n,d)$.

The bivariate viewpoint, i.e., the separation of $n$ and the degree $d$,
is also standard for every notion of a \emph{rank} of a tensor or polynomial, see for example~\cite{BL13}.

\paragraph{Main results: The bivariate viewpoint in linear algebra}
We show that the bivariate viewpoint appears naturally in linear algebra,
a priori outside of algebraic complexity theory, but with direct consequences such as $\w(\cpc_{n,d})\in\poly(n,d)$.
For this purpose, we prove our first main result, the bivariate Cayley--Hamilton theorem \ref{thm:CayleyHamilton}.
For a polynomial $f$ in $n^2$ variables $x_{i,j}$, $1\leq i,j\leq n$, let $\nabla f$ denote the $n \times n$ matrix whose entry at position $(i,j)$ is the partial derivative $\frac{\partial f} {\partial x_{i,j}}$, which we also denote by $\partial_{i,j}(f)$.
Note that $\nabla\det_n$ is the cofactor matrix of $X_n$, and its transpose $(\nabla\det_n)^\transpose$ is by definition the so-called adjugate matrix of $X_n$.

Our bivariate Cayley--Hamilton theorem \ref{thm:CayleyHamilton} states that for all $n,d$ we have
\begin{equation}
\tag{\protect{\textup{Thm.~\ref{thm:CayleyHamilton}}}}
(\nabla \cpc_{n,d+1})^\transpose \ = \ \sum_{i=0}^d (-1)^{i} \, \cpc_{n,d-i} \, X_n^i.
\end{equation}
Our theorem has several immediate corollaries, in particular the classical Cayley--Hamilton theorem, the Girard--Newton identities, and $\w(\cpc(n,d))\in\poly(n,d)$,
as depicted in Figure~\ref{fig:implicationstheorems}.
The proofs of all of the corollaries are very short, and we prove them all in \S\ref{sec:corollaries}.

\begin{figure}
\centering
\begin{tikzpicture}[scale=0.95]
\node[draw,rectangle,rounded corners=6pt] (bivCayHam) at (-2,0) {bivariate Cayley--Hamilton Thm~\ref{thm:CayleyHamilton}};
\node[draw,rectangle,align=center,rounded corners=6pt] (traceCayHam) at (0,-2) {trace Cayley--Hamilton\\$=$ matrix Girard--Newton};
\node[draw,rectangle,rounded corners=6pt] (GirardNewton) at (0,-4) {Girard--Newton identities};
\node[draw,rectangle,align=center,rounded corners=6pt] (adjugate) at (-9,-2) {Fadeev--LeVerrier\\adjugate};
\node[draw,rectangle,rounded corners=6pt] (CayHam) at (-4.85,-2) {Cayley--Hamilton Thm};
\node[draw,rectangle,rounded corners=6pt] (detinpoly) at (5,-2) {$\w(\cpc_{n,d})\in\poly(n,d)$};
\node[draw,rectangle,align=center,rounded corners=6pt] (detinpolycharzero) at (5,-4) {$\w(\cpc_{n,d})\in\poly(n,d)$\\over fields of char.\ 0};
\draw[double equal sign distance, -implies, shorten >=2pt, shorten <=2pt] (bivCayHam) -- (adjugate) node [midway,fill=white,inner sep=1pt] {\footnotesize Cor.~\ref{cor:adjugate}};
\draw[double equal sign distance, -implies, shorten >=2pt, shorten <=2pt] (bivCayHam) -- (traceCayHam) node [midway,fill=white,inner sep=1pt] {\footnotesize Cor.~\ref{cor:traceCayHam}};
\draw[double equal sign distance, -implies, shorten >=2pt, shorten <=2pt] (traceCayHam) -- (GirardNewton) node [midway,fill=white,inner sep=1pt] {\footnotesize Cor.~\ref{cor:girardnewton}};
\draw[double equal sign distance, -implies, shorten >=2pt, shorten <=2pt] (detinpoly) -- (detinpolycharzero);
\draw[double equal sign distance, -implies, shorten >=5pt, shorten <=5pt] (traceCayHam) -- (detinpolycharzero) node [midway,fill=white,inner sep=1pt] {\footnotesize Cor.~\ref{cor:wdetpolycharzero}};
\draw[double equal sign distance, -implies, shorten >=5pt, shorten <=5pt] (bivCayHam) -- (CayHam) node [midway,fill=white,inner sep=1pt] {\footnotesize Cor.~\ref{cor:CayHam}};
\draw[double equal sign distance, -implies, shorten >=5pt, shorten <=5pt] (bivCayHam) -- (detinpoly) node [midway,fill=white,inner sep=1pt] {\footnotesize Pro.~\ref{pro:wdetpoly}};
\draw[double equal sign distance, -implies, shorten >=15pt, shorten <=10pt, dashed]
  ( $(GirardNewton) + (-1.5cm,0)$ ) -- ( $(traceCayHam) + (-1.5cm,0)$ )
  node [pos=0.43,fill=white,inner sep=1pt] {\footnotesize \S\ref{subsec:algclosed}};
\end{tikzpicture}
\caption{Direct implications between the theorems.
The dashed implication only works over algebraically closed fields: one gets the matrix Girard--Newton identities from the classical Girard--Newton identities, and from there one derives $\w(\cpc_{n,d})\in\poly(n,d)$, but only over fields of characteristic zero.
To obtain $\w(\cpc_{n,d})\in\poly(n,d)$ over arbitrary commutative rings, one either uses our bivariate Cayley--Hamilton theorem, or one has to use an entirely separate construction such as \cite{Tod92,Val92,MV97}, see~\S\ref{sec:relwork}.
}
\label{fig:implicationstheorems}
\end{figure}
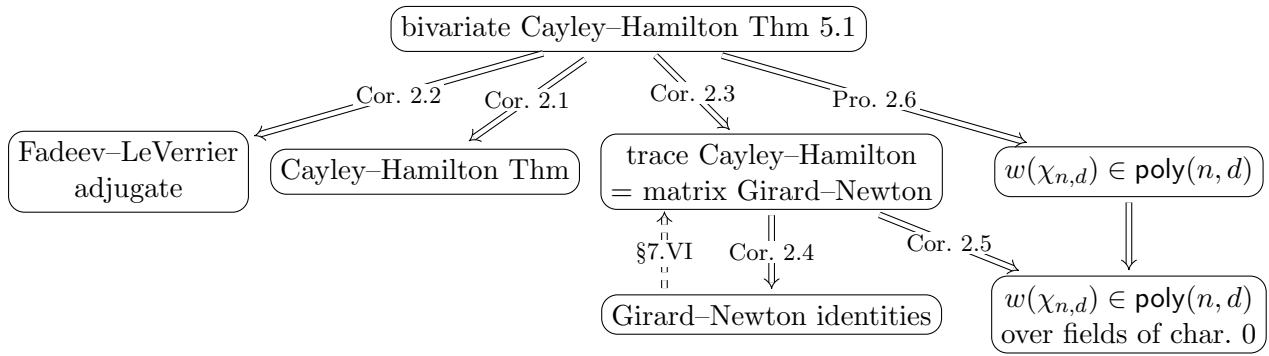

We study $(\nabla\cpc_{n,d})^\transpose$ more closely in \S\ref{sec:optimized}.
As our second main result,
we construct an ABP for computing all $\cpc_{n,d}$.
All intermediate computations results of our ABP are entries of $(\nabla\cpc_{n,d})^\transpose$.
Our construction is one third the size and width of the current smallest ABP \cite{MV97}, see \S\ref{sec:optimized}.
This is a modest improvement, but it marks the first improvement on this problem for 28 years.

Our ABP only ever computes the entries of $(\nabla\cpc_{n,d})^\transpose$ and nothing else,
and does not use any ad-hoc combinatorial constructions.
In particular, it does not involve the concept of clow sequences,
which are combinatorial gadgets that have no known nice algebraic interpretation, see \S\ref{subsec:clowsequences}.
Therefore, our construction answers the question raised in \S4 of \cite{MV99} about the necessity of clow sequences
for these types of algorithms
in the negative.
This answer to this 26-year-old open question is our third main result.

Some of our results are algebraic in nature, and our algorithmic results follow directly from algebraic identities.
We prove all results via the combinatorial approach to matrix algebra pioneered by \cite{Foa65,CF69,Str83,Zei85}, see many more references in \cite{Zei85,MV99}.
The construction in \cite{MV97} is combinatorial,
but not in exactly the same way as \cite{Str83,Zei85},
who give combinatorial proofs for algebraic identities between inherently algebraic objects,
while \cite{MV97} work with the combinatorial gadget of clow sequences.
Our proofs are in the spirit of \cite{Str83,Zei85}, since we
give combinatorial proofs of algebraic identities between objects from algebra.
In our case these are the first order partial derivatives of $\cpc_{n,d}$, which appear in the bivariate Cayley--Hamilton theorem.

\section{Corollaries}
\label{sec:corollaries}
In this section we prove the corollaries of Theorem~\ref{thm:CayleyHamilton} illustrated in Figure~\ref{fig:implicationstheorems}.
For a matrix $X$ we write $[X]_{i,j}$ for the entry in row $i$ and column $j$.

\begin{corollary}[The Cayley--Hamilton theorem]
\label{cor:CayHam}
\[
0 \ = \ \sum_{i=0}^n (-1)^{i} \, \cpc_{n,n-i} \, X_n^i.\qedhere
\]
\end{corollary}
\begin{proof}
Set $n=d$ in Theorem~\ref{thm:CayleyHamilton}.
Since $\cpc_{n,d}=0$ for $d>n$, it follows that for $n=d$ the left-hand side of Theorem~\ref{thm:CayleyHamilton} vanishes,
which finishes the proof.
\end{proof}

The Fadeev--LeVerrier algorithm computes the adjugate matrix along the way.
We will not discuss this algorithm here, because it uses divisions and hence does not work over arbitrary commutative rings,
but the following formula is used to prove that it computes the adjugate, see \protect{\cite[IV(46)]{Gan59}}.
\begin{corollary}[The adjugate matrix in the Fadeev--LeVerrier algorithm]
\label{cor:adjugate}
\[
(\nabla \det_n)^\transpose \ = \ \sum_{i=0}^{n-1} (-1)^{i} \, \cpc_{n,n-1-i} \, X_n^i.\qedhere
\]
\end{corollary}
\begin{proof}
We set $d=n-1$ in Theorem~\ref{thm:CayleyHamilton}.
%
%
\end{proof}

\begin{corollary}[The trace Cayley--Hamilton theorem, also known as the matrix Girard--Newton identities]
\label{cor:traceCayHam}
\[
-d \cdot \cpc_{n,d} \ = \ \sum_{i=1}^d (-1)^{i} \, \cpc_{n,d-i} \, \tr(X_n^i).
\qedhere
\]
\end{corollary}
\begin{proof}
Note that by \eqref{eq:cpcS} we have
\[
\partial_{a,a}(\cpc_{n,d+1}) \ = \ \sum_{a \notin S\in\binom{\{1,\ldots,n\}}{d}} \det[X_n]_{S,S}.
\]
Since for every $S\in\binom{\{1,\ldots,n\}}{d}$ there are $n-d$ many $a\in\{1,\ldots,n\}\setminus S$, it follows that
$\tr(\nabla \cpc_{n,d+1}) = (n-d) \cpc_{n,d}$.
Hence, taking the trace in Theorem~\ref{thm:CayleyHamilton} on both sides gives
$(n-d) \cpc_{n,d} = \sum_{i=0}^d (-1)^{i} \cpc_{n,d-i} \tr(X_n^i)$.
We conclude the proof by subtracting $n \, \cpc_{n,d}$ on both sides of the equation.
\end{proof}
%
%
%
\begin{corollary}[The Girard--Newton identities]
\label{cor:girardnewton}
Let $e_{n,d} = \cpc_{n,d}(\diag(x_1,\ldots,x_n))$ be the elementary symmetric polynomial ($e_{n,0}=1$).
Let $p_{n,d} := x_1^d+\cdots+x_n^d$ be the power sum polynomial ($p_{n,0}=n$).
We have
\[
-d \cdot e_{n,d} \ = \ \sum_{i=1}^d (-1)^{i} \, e_{n,d-i} \, p_{n,i}.\qedhere
\]
\end{corollary}
\begin{proof}
We restrict Corollary~\ref{cor:traceCayHam} to diagonal matrices $D=\diag(x_1,\ldots,x_n)$ and immediately obtain the desired statement, because $\cpc_{n,d}(D)=e_{n,d}$ and $\tr(D^d)=p_{n,d}$.
\end{proof}

Recall from \S\ref{sec:intro} that an ABP $G$
computes
$\sum_{\textup{$s$-$t$-path $p$ in $G$}} \beta(p)$,
where $\beta(p)=\prod_{e\in p}\beta(e)$,
and $\beta(e)$ is the edge label of the edge $e$.
A sub-ABP $G_v$ of $G$ at a vertex $v$ is defined as the ABP that consists of
all edges on all paths from $s$ to $v$,
where $v$ is the sink of $G_v$, and $s$ is the source of~$G_v$.
A polynomial $f$ is defined to be computed by an ABP $G$ \emph{along the way} if
there exists a vertex $v$ in $G$ such that $G_v$ computes~$f$.
We also say that $G$ computes $f$ at~$v$.
Sub-ABPs can also have sources other than $s$, and we see those in the constructions of
Corollary~\ref{cor:wdetpolycharzero} and Proposition~\ref{pro:wdetpoly}.

\begin{corollary}
\label{cor:wdetpolycharzero}
Over fields of characteristic zero, $\w(\cpc_{n,d})\in\poly(n,d)$ follows directly from the trace Cayley--Hamilton theorem (Corollary~\ref{cor:traceCayHam}).
\end{corollary}
\begin{proof}
We construct an ABP without constant edges.
The ABP will compute $\det_n=\cpc_{n,n}$,
and it will also compute each $\cpc_{n,d}$, $0 \leq d\leq n$ along the way, at a vertex that we call $v_d$.
See Figure~\ref{fig:charpolycharzero} for an illustration.
Starting with the vertices $v_0, \ldots, v_n$,
we connect each $v_d$ with all $v_{d-i}$, $1\leq i \leq d\leq n$, by ABPs $G_{d-i,d}$ that compute $\frac{(-1)^{i}}{-d}\tr(X_n^{i})$
via identifying the source of $G_{d-i,d}$ with $v_{d-i}$ and the sink of $G_{d-i,d}$ with $v_{d}$.
Note that each connection maintains the property that at each vertex layer $j$ the ABP computes homogeneous degree $j$ polynomials.
The resulting program has $\binom{n+1}{2}$ such connections, each of width $\w(\tr(X_n^i))\leq n^2$.
By Corollary~\ref{cor:traceCayHam}, the ABP computes $\cpc_{n,d}$ at each vertex~$v_d$.
Hence, $\w(\cpc_{n,d})\leq \binom{n+1}{2}\cdot n^2$ over fields of characteristic~0.
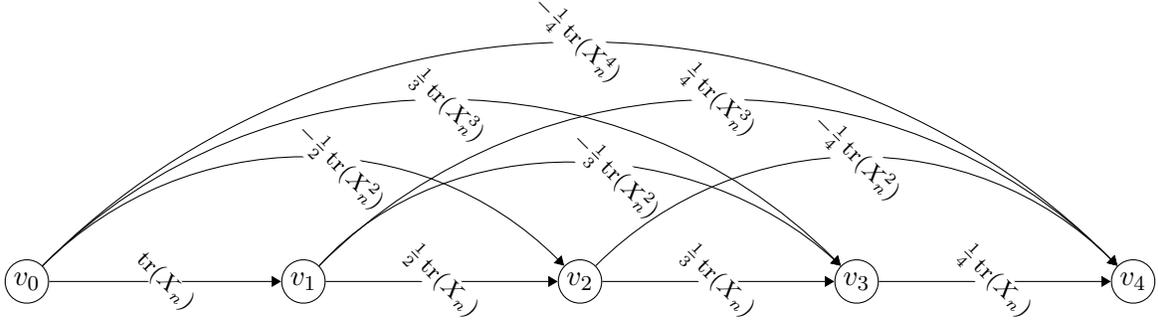
\begin{figure}
\centering
\tikzset{rect style/.style={rotate=#1,scale=2,every node/.style={rotate=#1}}}
\begin{tikzpicture}[rect style=-45,xscale=1.3,yscale=1.3]
\node[draw,circle,inner sep=0,minimum size=0.5cm] (v0) at (0,0) {\rotatebox{45}{$v_0$}};
\node[draw,circle,inner sep=0,minimum size=0.5cm] (v1) at (1,1) {\rotatebox{45}{$v_1$}};
\node[draw,circle,inner sep=0,minimum size=0.5cm] (v2) at (2,2) {\rotatebox{45}{$v_2$}};
\node[draw,circle,inner sep=0,minimum size=0.5cm] (v3) at (3,3) {\rotatebox{45}{$v_3$}};
\node[draw,circle,inner sep=0,minimum size=0.5cm] (v4) at (4,4) {\rotatebox{45}{$v_4$}};
\draw[-Triangle] (v0) to node[midway,fill=white,inner sep=0] {\footnotesize$\tr(X_n)$} (v1);
\draw[-Triangle] (v1) to node[midway,fill=white,inner sep=0] {\footnotesize$\frac 1 2\tr(X_n)$} (v2);
\draw[-Triangle] (v2) to node[midway,fill=white,inner sep=0] {\footnotesize$\frac 1 3\tr(X_n)$} (v3);
\draw[-Triangle] (v3) to node[midway,fill=white,inner sep=0] {\footnotesize$\frac 1 4\tr(X_n)$} (v4);
\draw[-Triangle] (v0) to[bend left=45] node[pos=0.55,fill=white,inner sep=0,minimum size=0.5cm] {\footnotesize$ \ \ \ -\frac 1 2\tr(X_n^2)$} (v2);
\draw[-Triangle] (v1) to[bend left=45] node[pos=0.55,fill=white,inner sep=0,minimum size=0.5cm] {\footnotesize$ \ \ \ \ -\frac 1 3\tr(X_n^2)$} (v3);
\draw[-Triangle] (v2) to[bend left=45] node[midway,fill=white,inner sep=0,minimum size=0.5cm] {\footnotesize$-\frac 1 4\tr(X_n^2)$} (v4);
\draw[-Triangle] (v0) to[bend left=45] node[midway,fill=white,inner sep=0,minimum size=0.5cm] {\footnotesize$ \ \ \frac 1 3\tr(X_n^3)$} (v3);
\draw[-Triangle] (v1) to[bend left=45] node[midway,fill=white,inner sep=0,minimum size=0.5cm] {\footnotesize$\frac 1 4\tr(X_n^3)$} (v4);
\draw[-Triangle] (v0) to[bend left=45] node[midway,fill=white,inner sep=0,minimum size=0.5cm] {\footnotesize$\!\!-\frac 1 4\tr(X_n^4)$} (v4);
\end{tikzpicture}
\caption{Over fields of characteristic 0. An ABP computing along the way all $\cpc_{n,d}$ for $d\leq 4$, where any $n\in\IN$ is fixed. Each edge represents a sub-ABP whose source is some $v_i$ and whose sink is some $v_j$, $i<j$.}
\label{fig:charpolycharzero}
\end{figure}
\end{proof}

In order to get the same result over arbitrary commutative rings,
we use Theorem~\ref{thm:CayleyHamilton}
instead of Corollary~\ref{cor:traceCayHam}
and make a bivariate argument, see the following proposition.
\begin{proposition}\label{pro:wdetpoly}
Over arbitrary commutative rings, we have
$\w(\cpc_{n,d}) \in\poly(n,d)$.
\end{proposition}
\begin{proof}
Let $\pow_{n,i} := [X_n^i]_{n,n}$ be the bottom-right entry of the $i$-th matrix power of $X_n$, and note $\w(\pow_{n,i}) \leq n$.
From \eqref{eq:cpcS} it follows that $[\nabla \cpc_{n,d+1}]_{n,n} = \cpc_{n-1,d}$.
We take the bottom-right matrix entry (i.e., position $(n,n)$) on both sides of Theorem~\ref{thm:CayleyHamilton}:
\begin{equation}
\label{eq:Samuelson}
\cpc_{n-1,d} \ = \ \sum_{i=0}^d (-1)^{i} \, \cpc_{n,d-i} \, \pow_{n,i}.
\end{equation}
We use that $\pow_{n,0}=1$ to rewrite equation~\eqref{eq:Samuelson} in the form
\begin{equation}
\label{eq:cpcrecurse}
\cpc_{n,d} \ = \ \cpc_{n-1,d} + \sum_{i=1}^{d} (-1)^{i+1} \, \cpc_{n,d-i} \, \pow_{n,i}.
\end{equation}
%
Equation \eqref{eq:cpcrecurse} gives a recipe to construct an ABP,
illustrated in Figure~\ref{fig:charpolyprogram}.
Our ABP will compute $\cpc_{n,d}$ at vertex $v_{n,d}$.
\begin{figure}
\centering
\begin{tikzpicture}[xscale=2.5,yscale=1.9]
\node[fill=white,draw,circle,inner sep=0,minimum size=0.6cm] (v40) at (0,-1) {$v_{0,0}$};
\node[fill=white,draw,circle,inner sep=0,minimum size=0.6cm] (v30) at (0,0) {$v_{1,0}$};
\draw[-Triangle] (v40) to node[left,pos=0.7,inner sep=0,minimum size=0.5cm] {\footnotesize$\ \ 1$} (v30);
\node[fill=white,draw,circle,inner sep=0,minimum size=0.6cm] (v20) at (0,1) {$v_{2,0}$};
\node[fill=white,draw,circle,inner sep=0,minimum size=0.6cm] (v10) at (0,2) {$v_{3,0}$};
\draw[-Triangle] (v30) to node[left,pos=0.7,inner sep=0,minimum size=0.5cm] {\footnotesize$\ \ 1$} (v20);
\draw[-Triangle] (v20) to node[left,pos=0.7,inner sep=0,minimum size=0.5cm] {\footnotesize$\ \ 1$} (v10);
\node[fill=white,draw,circle,inner sep=0,minimum size=0.5cm] (v11) at (1,0) {$v_{1,1}$};
\draw[-Triangle] (v30) to node[midway,fill=white,inner sep=1pt] {\footnotesize$\pow_{1,1}$} (v11);
\node[fill=white,draw,circle,inner sep=0,minimum size=0.5cm] (v21) at (1,1) {$v_{2,1}$};
\draw[-Triangle] (v20) to node[pos=0.5,fill=white,inner sep=0] {\footnotesize$\pow_{2,1}$} (v21);
\draw[-Triangle] (v11) to node[left,pos=0.7,inner sep=0,minimum size=0.5cm] {\footnotesize$\ \ 1$} (v21);
\node[fill=white,draw,circle,inner sep=0,minimum size=0.5cm] (v31) at (1,2) {$v_{3,1}$};
\draw[-Triangle] (v10) to node[pos=0.5,fill=white,inner sep=0] {\footnotesize$\pow_{3,1}$} (v31);
\draw[-Triangle] (v21) to node[left,pos=0.7,inner sep=0,minimum size=0.5cm] {\footnotesize$\ \ 1$} (v31);
\node[fill=white,draw,circle,inner sep=0,minimum size=0.6cm] (v22) at (2,1) {$v_{2,2}$};
\draw[-Triangle] (v20) to[bend left] node[pos=0.5,fill=white,inner sep=0] {\footnotesize$-\pow_{2,2}$} (v22);
\draw[-Triangle] (v21) to node[pos=0.5,fill=white,inner sep=1pt] {\footnotesize$\pow_{2,1}$} (v22);
\node[fill=white,draw,circle,inner sep=0,minimum size=0.6cm] (v32) at (2,2) {$v_{3,2}$};
\draw[-Triangle] (v31) to node[pos=0.5,fill=white,inner sep=1pt] {\footnotesize$\pow_{3,1}$} (v32);
\node[fill=white,draw,circle,inner sep=0,minimum size=0.6cm] (v33) at (3,2) {$v_{3,3}$};
\draw[-Triangle] (v32) to node[pos=0.5,fill=white,inner sep=1] {\footnotesize$\pow_{3,1}$} (v33);
\draw[-Triangle] (v22) to node[left,pos=0.7,inner sep=0,minimum size=0.5cm] {\footnotesize$\ \ 1$} (v32);
\draw[-Triangle] (v10) to[bend left,looseness=1] node[pos=0.5,fill=white,inner sep=0] {\footnotesize$-\pow_{3,2}$} (v32);
\draw[-Triangle] (v31) to[bend left,looseness=1] node[pos=0.5,fill=white,inner sep=0] {\footnotesize$-\pow_{3,2}$} (v33);
\draw[-Triangle] (v10) to[bend left=40] node[pos=0.5,fill=white,inner sep=0] {\footnotesize$\pow_{3,3}$} (v33);
\end{tikzpicture}
\caption{An ABP computing along the way all $\cpc_{n,d}$ for $d\leq n\leq 3$.
Each non-constant edge represents a sub-ABP.
%
}
\label{fig:charpolyprogram}
\end{figure}
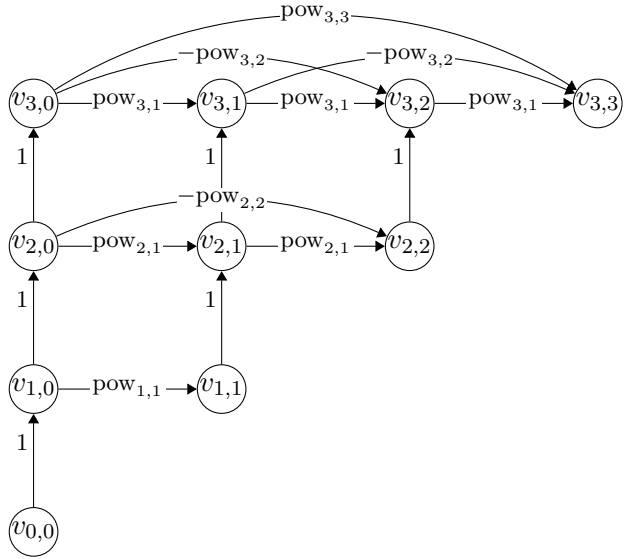
We start with a sequence of vertices $v_{i,0}$, $0\leq i\leq n$,
with edges labeled with 1 from each $v_{i,0}$ to $v_{i+1,0}$.
The source is $v_{0,0}$.
Equation~\eqref{eq:cpcrecurse} allows us to
take a program that computes $\cpc_{n-1,d}$ and each $\cpc_{n,d-i}$, $i\in\{1,\ldots,d\}$,
and enlarge it to a program that also computes $\cpc_{n,d}$ at a vertex that we call $v_{n,d}$, as follows.
We add an edge with label 1 from $v_{n-1,d}$ to $v_{n,d}$,
provided $n>d$ (otherwise, $\cpc_{n-1,d}=0$).
Then, for each $i\in\{1,\ldots,d\}$ we connect $v_{n,d-i}$ and $v_{n,d}$ via an ABP that computes $(-1)^{i+1}\,\pow_{n,i}$, identifying sources and sinks as in the proof of Corollary~\ref{cor:wdetpolycharzero}.
We enlarge the program via this construction several times
in the following order of pairs $(n,d)$ that guarantees that all required vertices are present at each step of the construction:
$(1,1), (2,1), \ldots, (n,1), (1,2), (2,2), \ldots, (n,2), (3,1), \ldots, \ldots, (n,d)$.
The final ABP consists of $O(nd)$ many named vertices $v_{i,j}$.
Every $v_{i,j}$ has exactly $j$ many $\pow_{i,j'}$, $1\leq j' \leq j$, ending at it,
which gives a total of $O(n d^2)$ many $\pow_{i,j'}$.
We have $\w(\pow_{i,j'})\leq n$ in every case, which results in $\w(\cpc_{n,d}) \in O(n^2 d^2)$.
\end{proof}
%
%
Figure~\ref{fig:charpolyprogram} is the first time
over arbitrary commutative rings
that a polynomially sized ABP for $\det_d$ can be easily depicted.
Its width is $O(n^2 d^2)$, which is surprisingly small given the fact that it is a direct application of Theorem~\ref{thm:CayleyHamilton},
which is first and foremost a theorem in linear algebra, not in algorithms.
The algorithmic and algebraic questions appear to be deeply interlinked here.

\section{Related work}\label{sec:relwork}

\paragraph{Girard--Newton, Cayley--Hamilton, and combinatorial proofs in algebra}
The Girard--Newton identities have numerous applications in constructions in algebraic complexity theory.
For example, recently the Girard--Newton identities have been applied in two breakthrough results in algebraic complexity theory \cite{LST25, AW24},
and they are the center of investigation in recent papers, both in
algebraic complexity theory \cite{FLST24} and algebra \cite{dVr25}.

The Cayley--Hamilton theorem is a fundamental theorem in linear algebra with countless applications and generalizations.
\cite{SC04} studies the Cayley--Hamilton theorem from a proof complexity perspective.
In algebraic complexity theory the Cayley--Hamilton theorem can be used to explain the symmetries of Strassen's fast matrix multiplication algorithm \cite{IL19}.
\cite{Gri25} first discusses the long history of the trace Cayley--Hamilton theorem,
the different proof difficulties for the cases $d\geq n$ and $d<n$, and the available proofs over fields.
\cite{Gri25} then states a lemma (Lemma 3.11(c) therein, already implicit in \cite[\S3]{Buc1884}) over arbitrary commutative rings that he then applies in the proofs of both the trace Cayley--Hamilton theorem and the Cayley--Hamilton theorem.
The Cayley--Hamilton theorem follows directly from this lemma, but more work is required to reach the trace Cayley--Hamilton theorem. Before our paper, \cite{Gri25} was the most unified view on the subject.
The formula on the right-hand side of Theorem~\ref{thm:CayleyHamilton} appears in earlier works, see for example \cite[IV(46)]{Gan59} or \cite[Lemma~3.11(c)]{Gri25}, but the central importance of the matrix on the left-hand side was overlooked
and it was left underanalysed with only some of its properties proved, for example its trace.
The matrix on the left-hand side was only known in the special case of $d=n$ (Corollary~\ref{cor:CayHam}) and $d=n-1$ (Corollary~\ref{cor:adjugate}).

Subtleties in proving the Girard--Newton identities are explained in \cite{Min03}.
A short combinatorial proof of the Girard--Newton identities is presented in \cite{Zei84},
and a combinatorial proof of the Cayley--Hamilton theorem is presented in \cite{Rut64,Str83}.
Numerous other algebraic results have been proved combinatorially,
for example the MacMahon master theorem \cite{Foa65},
the matrix-tree theorem \cite{Orl78,Tem81,Cha82,Min97},
the Jacobi identity \cite{Foa80},
Vandermonde's determinant identity \cite{Ges79}, and others.
These results are listed in \cite{MV99} and surveyed in \cite{Zei85}, \cite[Ch.~4]{Sta12}.

\paragraph{ABPs, GapL, and clow sequences}
There are many ways for computing $\det_d$ efficiently,
but not all of them directly give a proof that $\w(\det_d)\in\poly(d)$
over arbitrary commutative rings.
\cite{Coo85} noted that the Samuelson--Berkowitz algorithm \cite{Sam42,Ber84}
can be used to obtain $\det_d$ efficiently from matrix powers,
but this does not directly prove $\w(\det_d)\in\poly(d)$.
It is stated in \cite{Tod91} that this observation in \cite{Coo85}
can be converted into a proof that $\w(\det_d)\in\poly(d)$.
A proof is provided in \cite{Tod92}.
This is also proved in \cite[\S3]{Val92}.
These proofs also directly give $\w(\cpc_{n,d})\in\poly(n,d)$.
Another proof is provided in \cite{MV97,MV97SODA},
which uses a combinatorial construction and gives the smallest ABP known so far.
All these constructions have natural Boolean counterparts,
and they can be used to show that the integer determinant computation is complete for the counting complexity class $\GapL$, see \cite{AO96}.
This is a surprising fact, because $\GapL$ is not defined in an algebraic way.
The class $\GapL$ consists of all functions $f:\{0,1\}^*\to\IZ$ for which there exist two nondeterministic logspace Turing machines $M_+$ and $M_-$ such that $\forall w\in\{0,1\}^*$ we have $f(w)=\#\textup{acc}_{M_+}(w)-\#\textup{acc}_{M_-}(w)$, where
$\#\textup{acc}_M(w)$ is the number of accepting computation paths of $M$ on input $w$.

See \cite[\S2.8]{Rot01} for a discussion of the history and attribution of the Samuelson--Berkowitz algorithm.

\cite{MV97} prove $\w(\cpc_{n,d})\in\poly(n,d)$ via the cancellation properties of so-called \emph{clow sequences}, see \S\ref{subsec:clowsequences} below.
Clow sequences and their cancellations have also been successfully used in \cite{ARZ99}.
\cite{Sol02} finds that the Samuelson--Berkowitz algorithm can be described naturally in terms of clow sequences.
\cite{MV99} analyze several algorithms \cite{Sam42,Csa75,Chi85,MV97} and their internal use of clow sequences
and come to the conclusion
that all published efficient algorithms for $\cpc_{n,d}$
that are based on signed weighted sums of partial cycle covers
use the concept of clow sequences or a generalization of clow sequences.
The ABP that we construct in Theorem~\ref{thm:wdetnsquare}
computes only first order partial derivatives of partial cycle covers, and does not involve the more general concept of clow sequences (partial cycle covers are clow sequences, but not vice versa).
Hence, our construction answers the question raised in \cite[\S4]{MV99} about the necessity of the use of clow sequences of such algorithms in the negative, see also \S\ref{subsec:clowsequences}.

\paragraph{Structure, homogenization, and geometric complexity theory} The width $\w(f)$ can be defined as the smallest $n$ such that
$f$ is restriction of the iterated matrix multiplication polynomial
\begin{equation*}
\IMM_{n,d} \ = \ \sum_{1 \leq i_1,\ldots,i_{d-1} \leq n} y^{(1)}_{n,i_1} \, y^{(2)}_{i_1,i_2} \, y^{(3)}_{i_2,i_3} \, \cdots \, y^{(d-1)}_{i_{d-2},i_{d-1}} \, y^{(d)}_{i_{d-1},n},
\end{equation*}
where a restriction is defined as a map that sends variables to linear combinations of variables.
The polynomial $\IMM_{n,d}$ has been recently the object of study in the breakthrough result \cite{LST25}.
The high amount of structure of $\IMM_{n,d}$ makes it not only possible to efficiently sample random polynomials of fixed width and degree, but also to analyze their probabilistic properties and obtain algorithms for solving systems of random polynomial equations of given width \cite{BCL23},
which is a solution to Smale's 17th problem.

Early papers in geometric complexity theory (GCT) use the determinantal complexity $\dc$ as their main complexity measure, and \eqref{eq:wdc} indicates that width and $\dc$ should work equally well from a computational complexity perspective, but since the determinant is a single-index polynomial, it is necessary to either discard the idea that (border) complexity $\leq n$ polynomials form a $\GL$ orbit closure, or one must introduce \emph{padding}, see \cite{MS01,MS08,BLMW11}.
However, padding has serious effects on the representation theoretic multiplicities \cite{KL14},
which led to main GCT conjectures being disproved in \cite{IP17,BIP19}. These proofs rely heavily on the padding, and if the width is used instead of $\dc$, then no such negative results are known.
\cite{DGIJL24} introduce the noncommutative elementary symmetric polynomial evaluated at $3\times 3$ variable matrices as a simpler replacement for $\IMM_{n,d}$, which makes the study of small finite cases easier.
This polynomial is also double-indexed, but it is not characterized by its stabilizer. That polynomial was recently used in \cite{FLST24} to study the homogenization of arithmetic formulas.

\section{Preliminaries}\label{sec:prelim}
For a function $\zeta:\IN\times\IN\to\IN$ we write $\zeta\in\poly(n,d)$ if there exists a bivariate polynomial $p\in\IZ[n,d]$ such that for all $n\geq 1$ and $d\geq 1$ we have $\zeta(n,d)\leq p(n,d)$.

All graphs in this paper are directed graphs, or \emph{digraphs} for short, without multi-edges, but we allow loops (i.e., edges from a vertex to itself).
Formally,
a \emph{digraph} $G$ consists of a \emph{vertex set} $V(G)$ and an \emph{edge set} $E(G)\subseteq V(G)\times V(G)$.
A \emph{walk} $q$ of \emph{length} $l\geq 0$ in $G$ is a sequence of vertices $(v_0, v_1, \ldots, v_l)$
for which $\forall i\in\{1,\ldots,l\}: (v_{i-1},v_{i}) \in E(G)$.
The edge set of a walk is defined as $\{(v_{i-1},v_{i}) \mid 1 \leq i \leq l\}$.
The sign $\sgn(q)$ is defined as $(-1)^{l}$.
For $s,t\in V(G)$ let $W_l(G,s,t)$ denote the set of length $l$ walks in $G$ from vertex $s=v_0$ to vertex~$t=v_l$.
A \emph{path} $p$ is a walk in which all $v_i$ are pairwise distinct.
In particular, a path never uses a loop.
A path can consist of a single vertex.
The vertex $v_0$ is called the tail of $p$, and $v_l$ is called the head of $p$.
Let $P(G,s,t)$ denote the set of paths in $G$ from vertex $s$ to vertex $t$.
Let $P_l(G,s,t)\subseteq P(G,s,t)$ denote the subset of length $l$ paths.
A \emph{cycle} is a set of edges that is obtained from the edge set of a path $p$ by adding the edge from the head of $p$ to its tail.
In particular, a loop is a cycle.
The length of a cycle is defined as its number of edges.
The sign of a cycle $c$ is defined as the sign of a path that is obtained by removing an edge from~$c$.
We write $(\!(v_0,v_1,\ldots,v_l)\!)$ for the cycle that has the edge set $\{(v_0,v_1),(v_1,v_2),\ldots,(v_{l-1},v_l),(v_l,v_0)\}$, and we write $(\!(v_0)\!)$ for the cycle formed by the loop at vertex~$v_0$.
For a walk $q$ in $G$ and $e \in E(G)$ we write $e\in q$ if $e$ is an edge in $q$.
Similarly, for $v\in V(G)$, we write $v \in q$ if $v$ is a vertex of~$q$.

We will mostly consider \emph{labeled digraphs}, which are pairs $(G,\beta_G)$ of a digraph $G$ and a label function $\beta_G:E(G)\to R[\mathbf x]$.
For an edge $e\in E(G)$ of a labeled digraph $(G,\beta_G)$, we define the weight $\beta(e) := \beta_G(e)$.
We usually do not give an explicit name to the label function, and we write ``$G$ is a labeled digraph'', and we access the weight via $\beta(e)$.
For a walk or cycle $q$ on a labeled digraph $G$, define the weight $\beta(q)$ of $q$ as
\[
\textstyle\beta(q) \ := \ \prod_{e \in q} \beta(e).
\]
Labeled digraphs are drawn with their labels on the edges,
where edges with label 0 are usually not drawn.

In a digraph $G$, a set of pairwise vertex-disjoint cycles is called a \emph{partial cycle cover}.
The sign of a partial cycle cover is defined as the product of the signs of its cycles (1 if it is the empty partial cycle cover).
The length $\ell(q)$ of a partial cycle cover $q$ is the sum of the lengths of its cycles (0 if it is the empty partial cycle cover).
Let $\cyco(G)$ denote the set of partial cycle covers of~$G$,
and let $\cyco_d(G)$ denote the set of partial cycle covers of $G$ of length~$d$.
Let $K_n$ denote the complete digraph on $n$ vertices, i.e., the digraph with vertex set $\{v_1,\ldots,v_n\}$ and edge set $\{(v_i,v_j)\mid 1\leq i,j\leq n\}$.
We attach edge labels to $K_n$ via $\beta(v_i,v_j):=x_{i,j}$, making $K_n$ a labeled digraph.
We use the short notation $\cyco = \cyco(K_n)$, $\cyco_d = \cyco_d(K_n)$, $P(v_a,v_b) = P(K_n,v_a,v_b)$, $W(v_a,v_b) = W(K_n,v_a,v_b)$, $V=V(K_n)$, and so on.
We have
\begin{equation}
\label{eq:detascyccov}
\textstyle\det_d = \sum_{q \in \cyco_d(K_d)} \sgn(q) \, \beta(q).
\end{equation}
By \eqref{eq:cpcS} it follows
\begin{equation}\label{eq:cpccyccov}
\textstyle\cpc_{n,d} = \sum_{q \in \cyco_d} \sgn(q) \, \beta(q).
\end{equation}
The set $\CCP_d(v_a,v_b)$
consists of all pairs $(q,z)$ for which
$z\in P(v_a,v_b)$, $q\in \cyco$, $\ell(z)+\ell(q)=d$, and $z$ and $q$ are vertex-disjoint.
$\CCP$ stands for (partial) cycle cover and path.

\subsection{Clow sequences}
\label{subsec:clowsequences}
The combinatorial concept of a clow sequence is defined in \cite{MV97} as follows,
and it is a bit more technically involved than what we have introduced so far.
We will not use this concept,
but we list its definition here to make it apparent that our construction answers the question in \cite[\S4]{MV99}.

A clow (CLOsed Walk) is a walk $q=(v_0,\ldots,v_l)$ on $G$ with $v_0=v_l$
such that the smallest entry in $\{v_1,\ldots,v_l\}$ appears exactly once.
This entry is called the head $h(q)$ of the clow.
A clow sequence $(q_1,q_2,\ldots,q_k)$ is a sequence of clows
such that $h(q_1)<h(q_2)<\cdots<h(q_k)$.
As noted in \cite{MV97}, a (partial) cycle cover is a clow sequence, but not vice versa.
In our construction Theorem~\ref{thm:wdetnsquare},
all polynomials computed along the way in our ABP
correspond to elements in some $\CCP$.
Since clow sequences can self-intersect, our construction only uses a small subset of the clow sequences.
It is reasonable to assume that some existing proofs in the literature that rely on the combinatorics of clow sequences simplify when phrasing them via $\CCP$ instead,
which are just partial cycle covers with one edge removed.

As elaborated in \S\ref{sec:relwork} in the context of GCT,
having constructions without ad-hoc combinatorial gadgets can be preferable,
especially when working with advanced tools from algebra and representation theory.
Our construction in Theorem~\ref{thm:wdetnsquare} uses only the first order partial derivatives of the $\cpc_{n,d}$,
which is an inherently algebraic notion.
No inherently algebraic notion is known to correspond to clow sequences.

\section{The Bivariate Cayley--Hamilton Theorem}
\label{subsec:genCayHam}

In this section we prove the bivariate Cayley--Hamilton theorem.
The proof is a combination of
the combinatorial proof of the Cayley--Hamilton theorem in \cite{Str83}
and
the combinatorial proof of the Girard--Newton identities in \cite{Zei84}.

\begin{theorem}[Bivariate Cayley--Hamilton Theorem]
\label{thm:CayleyHamilton}
For all $n,d\in \IN$ we have
\[
(\nabla \cpc_{n,d+1})^\transpose \ = \ \sum_{i=0}^d (-1)^{i} \, \cpc_{n,d-i} \, X_n^i.\qedhere
\]
\end{theorem}
\begin{proof}
We prove this matrix equation for each position $(a,b)$ independently, so fix $1\leq a,b\leq n$ and consider the entry at row $a$ and column~$b$. We write $[X]_{a,b}$ for the entry of the matrix $X$ at position $(a,b)$.

For the left-hand side of the claimed identity,
by \eqref{eq:cpccyccov} we have
\begin{equation}\label{eq:nabla}
[(\nabla \cpc_{n,d+1})^\transpose]_{a,b} \ = \ [\nabla \cpc_{n,d+1}]_{b,a} \ = \ \sum_{\substack{r \in \cyco_{d+1} \\ (v_b,v_a)\in r}} \sgn(r) \, \frac{\beta(r)}{x_{b,a}}
\end{equation}
(the division here and all divisions in this paper are just for notational convenience, not actual algorithmic division operations).

For the right-hand side of the claimed identity,
for every $i \in \{0,\ldots,d\}$ we have
\[
[(-1)^{i} X_n^i]_{a,b} = \sum_{z \in W_i(v_a,v_b)} \sgn(z) \, \beta(z).
\]
Therefore,
\begin{eqnarray}
\label{eq:sumqp}
\sum_{i=0}^d \cpc_{n,d-i} \cdot [(-1)^{i} X_n^i]_{a,b} & = & \sum_{i=0}^d \sum_{\substack{q \in \cyco_{d-i} \\ z \in W_i(v_a,v_b)}}
\sgn(q) \, \beta(q) \, 
\sgn(z) \, \beta(z)
\nonumber
\\
& = & 
\sum_{\substack{q \in \cyco \\ z \in W(v_a,v_b)\\\ell(q)+\ell(z)=d}}
\sgn(q) \, \beta(q) \, 
\sgn(z) \, \beta(z).
\end{eqnarray}
We now prove that the right-hand side of \eqref{eq:sumqp} equals the right-hand side of \eqref{eq:nabla},
which finishes the proof.
We first handle all summands in the right-hand side of \eqref{eq:nabla} and show how they appear in
the right-hand side of \eqref{eq:sumqp}.
We then show that the remaining summands in the right-hand side of \eqref{eq:sumqp} cancel out with each other.

For each $r \in \cyco_{d+1}$ with $(v_b,v_a)\in r$, we construct a pair $(q,z)\in\CCP_d(v_a,v_b)$
by removing the edge $(v_b,v_a)$ from $r$
(note that if $a=b$, then the loop $(v_a,v_a)$ is removed, and hence $z$ is the length 0 path from $v_a$ to $v_a$).
We observe that $\sgn(r)\beta(r)/x_{b,a}=\sgn(q)\beta(q)\sgn(z)\beta(z)$.
Every pair $(q,z)\in\CCP_d(v_a,v_b)$
is obtained from some $r$ in this way as follows.
From $(q,z)$ we construct $r$ by adding back the edge $(v_b,v_a)$.

It remains to show that all other summands in the right-hand side of \eqref{eq:sumqp} cancel out with each other.
These remaining pairs $(z,q)$, $q\in\cyco$, $z\in W(v_a,v_b)$, $\ell(q)+\ell(z)=d$, are exactly those for which
\begin{compactitem}
\item either $z$ uses a vertex at least twice (i.e., $z\in W(v_a,v_b)\setminus P(v_a,v_b)$),
\item or $z$ and $q$ are not vertex-disjoint,
\end{compactitem}
or both.
Let $X$ denote the set of these pairs $(z,q)$.
We will decompose $X$ into a disjoint union $X=\dot\bigcup_{v\in V} (S_v \mathop{\dot\cup} T_v)$ of $2n$ many sets $S_v$ and $T_v$.
For each $v$ we will give a weight-preserving sign-inverting bijection between $S_v$ and $T_v$.
This then finishes the proof.

Every $(z,q)\in X$ belongs to exactly one set $S_v$ or $T_v$
according to the following rule.
We start at the vertex $v_a$ and traverse the walk $z\in W(v_a,v_b)$ until we either use a vertex $v$ for the second time (this means that $(z,q) \in S_v$) or $z$ uses a vertex $v$ of $q$ (this means that $(z,q) \in T_v$), at which time we stop the traversal.
Note that each $(z,q)$ is contained in a unique set: if a vertex $v$ that is a vertex of $q$ is used a second time by~$z$, then this vertex of $q$ would have been detected in the traversal as a vertex of $q$ before visiting it a second time.

For $(z,q)\in S_v$, we see that $z$ starts with the vertices of a path $p$ from $v_a$ to $v$,
then continues with the vertices of a cycle $c$ from $v$ to $v$, and then $z$ continues with the vertices of a walk $t$ from $v$ to $v_b$.
Note that $q$ and $c$ are vertex-disjoint, because otherwise the traversal of $z$ would have found their common vertex before reaching~$v$ a second time.
We define $z'$ as a modification of $z$ as follows.
The walk $z'$ starts with $p$ and then directly continues with $t$ without taking the cycle $c$.
We define $q'$ to be the union of $q$ and~$c$.
The situation is illustrated in Figure~\ref{fig:CH}.
We observe that $(z',q')\in T_v$ by construction.
Note $\ell(q')+\ell(z')=d$. We have $\sgn(q)=(-1)^{\ell(c)-1}\sgn(q')$ and $\sgn(z)=(-1)^{\ell(c)}\sgn(z')$,
hence $\sgn(q)\,\sgn(z) = -\sgn(q')\,\sgn(z')$ as desired.
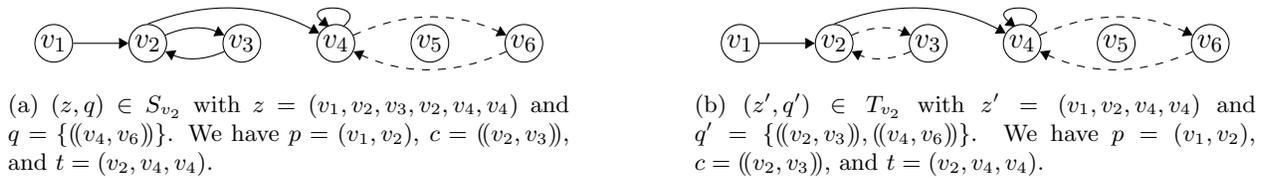
\begin{figure}[htbp]
    \centering
    \begin{subfigure}[b]{0.45\textwidth}
        \centering
\begin{tikzpicture}[xscale=1.25]
\node[draw,circle,inner sep=0,minimum size=0.5cm] (1) at (1,0) {$v_1$};
\node[draw,circle,inner sep=0,minimum size=0.5cm] (2) at (2,0) {$v_2$};
\node[draw,circle,inner sep=0,minimum size=0.5cm] (3) at (3,0) {$v_3$};
\node[draw,circle,inner sep=0,minimum size=0.5cm] (4) at (4,0) {$v_4$};
\node[draw,circle,inner sep=0,minimum size=0.5cm] (5) at (5,0) {$v_5$};
\node[draw,circle,inner sep=0,minimum size=0.5cm] (6) at (6,0) {$v_6$};
\draw[-Triangle] (1) -- (2);
\draw[-Triangle] (2) to[bend left] (3);
\draw[-Triangle] (3) to[bend left] (2);
\draw[-Triangle] (2.north) to[bend left] (4);
\draw[-Triangle] ($(4.north)!0.5!(4.north east)$) to[out=45, in=135, looseness=8] ($(4.north)!0.5!(4.north west)$);
\draw[-Triangle,dashed] (4) to[bend left] (6);
\draw[-Triangle,dashed] (6) to[bend left] (4);
\end{tikzpicture}
        \caption{$(z,q)\in S_{v_2}$ with $z=(v_1,v_2,v_3,v_2,v_4,v_4)$ and $q = \{(\!(v_4,v_6)\!)\}$. We have $p=(v_1,v_2)$, $c = (\!(v_2,v_3)\!)$, and $t = (v_2,v_4,v_4)$.}
        \label{fig:CHsecond}
    \end{subfigure}
    \hfill
    \begin{subfigure}[b]{0.45\textwidth}
        \centering
\begin{tikzpicture}[xscale=1.25]
\node[draw,circle,inner sep=0,minimum size=0.5cm] (1) at (1,0) {$v_1$};
\node[draw,circle,inner sep=0,minimum size=0.5cm] (2) at (2,0) {$v_2$};
\node[draw,circle,inner sep=0,minimum size=0.5cm] (3) at (3,0) {$v_3$};
\node[draw,circle,inner sep=0,minimum size=0.5cm] (4) at (4,0) {$v_4$};
\node[draw,circle,inner sep=0,minimum size=0.5cm] (5) at (5,0) {$v_5$};
\node[draw,circle,inner sep=0,minimum size=0.5cm] (6) at (6,0) {$v_6$};
\draw[-Triangle] (1) -- (2);
\draw[-Triangle,dashed] (2) to[bend left] (3);
\draw[-Triangle,dashed] (3) to[bend left] (2);
\draw[-Triangle] (2.north) to[bend left] (4);
\draw[-Triangle] ($(4.north)!0.5!(4.north east)$) to[out=45, in=135, looseness=8] ($(4.north)!0.5!(4.north west)$);
\draw[-Triangle,dashed] (4) to[bend left] (6);
\draw[-Triangle,dashed] (6) to[bend left] (4);
\end{tikzpicture}
        \caption{$(z',q')\in T_{v_2}$ with $z'=(v_1,v_2,v_4,v_4)$ and $q' = \{(\!(v_2,v_3)\!),(\!(v_4,v_6)\!)\}$. We have $p=(v_1,v_2)$, $c=(\!(v_2,v_3)\!)$, and $t=(v_2,v_4,v_4)$.}
        \label{fig:CHcommon}
    \end{subfigure}
    \caption{The situation in the proof of Theorem~\ref{thm:CayleyHamilton}.
The left and right subfigure are partner summands in the sign-inverting involution.
The figures are adaptions of figures in \cite{Str83}.}
    \label{fig:CH}
\end{figure}

For $(z',q')\in T_v$, the walk $z'$ starts with the vertices of a path $p$ from $v_a$ to $v$,
where $v$ is part of a cycle $c$ of $q'$,
and then $z'$ continues with the vertices of a walk $t$ from $v$ to $v_b$.
We define $z$ to start with the vertices of the path $p$, then continue with the vertices of the cycle $c$, then continue with the walk $t$.
We define $q$ to be $q'$ with $c$ removed.
The situation is illustrated in Figure~\ref{fig:CH}.
We observe that $(z,q)\in S_v$ by construction.
Note $\ell(q)+\ell(z)=d$. We have $\sgn(q)=(-1)^{\ell(c)-1}\sgn(q')$ and $\sgn(z)=(-1)^{\ell(c)}\sgn(z')$,
hence $\sgn(q)\,\sgn(z) = -\sgn(q')\,\sgn(z')$ as desired.
\end{proof}

\section{The entries of the gradient}
\label{sec:optimized}
Theorem~\ref{thm:CayleyHamilton} highlights the importance of the gradient transpose, $(\nabla\cpc_{n,d})^\transpose$.
We analyze the entries of $(\nabla\cpc_{n,d})^\transpose$ more closely in this section.

For a finite set of polynomials $S$ we define $\w(S)$ to be the width of an ABP that computes every element of $S$ along the way.
Such an ABP is allowed to have several sinks in vertex layer $d$.
The number of vertices for the ABP computing $\{\cpc_{n,d}\mid d\leq n\}$ in \cite{MV97,MV97SODA} (excluding source and sinks) is $n^3$, and the width is $n^2$ (see the discussion about pruning in \cite[\S5]{MV97}).
With a similar technique as in Proposition~\ref{pro:wdetpoly} and a bit more care
we prove the following Theorem~\ref{thm:wdetnsquare}, which improves the leading coefficients of the size by a factor of 3 and width by a factor of 2 compared to \cite{MV97}.
Instead of clow sequences, we use cycle-cover-and-paths ($\CCP$).
\begin{theorem}
\label{thm:wdetnsquare}
$\w(\{\cpc_{n',d'}\mid n'\leq n, \, d'\leq d\})\leq \frac{n^2}{2} + \textup{l.o.t.}$
The total number of vertices in this construction is $\frac{dn^2}{2}-\frac{d^3}{6} + \textup{l.o.t.}$,
which for $d=n$ gives $\frac{n^3}{3} + \textup{l.o.t.}$
\end{theorem}
\begin{proof}
Let $r_{n,d}$ be the last row of $(\nabla \cpc_{n,d+1})^\transpose$.
We decompose the $n\times n$ matrix $X_n$ into an $n \times (n-1)$ matrix $L_n$ and an $n\times 1$ matrix $C_n$, i.e.,
$X_n = (L_n|C_n)$.
The following block matrix identity is the main ingredient to build an ABP.
\begin{equation}
\label{eq:rnd}
r_{n,d} = \Big( -r_{n,d-1} L_n \,\Big|\, \sum_{i=d}^{n-1} r_{i,d-1} C_{i} \Big).
\end{equation}
We now prove \eqref{eq:rnd}, similarly to Theorem~\ref{thm:CayleyHamilton}.
For a vector $r$, let $[r]_a$ denote the $a$-th entry of $r$.
\begin{eqnarray*}
[r_{n,d}]_a & = & \sum_{\substack{q \in \cyco_{d+1}\\ (v_a,v_n) \in q}} \sgn(q) \frac{\beta(q)}{x_{a,n}}
 \ = \ 
\sum_{(q,z)\in\CCP_d(v_n,v_a)} \sgn(q) \, \beta(q) \, \sgn(z) \, \beta(z)
\end{eqnarray*}
In particular,
\begin{equation}\label{eq:rndncpc}
[r_{n,d}]_n = \cpc_{n-1,d}
\end{equation}
and
\begin{equation}
\label{eq:ridmCi}
r_{i,d-1} C_i \ = \ \sum_{a=1}^i \,
\sum_{\substack{q \in \cyco_d(K_i)\\ (v_a,v_i) \in q}} \sgn(q) \frac{\beta(q)}{x_{a,i}} \, x_{a,i}
 \ = \
\sum_{\substack{q \in \cyco_d(K_i)\\v_i \in q}} \sgn(q) \, \beta(q).
\end{equation}
For $1 \leq a < n$ we have
\begin{eqnarray}
[r_{n,d-1} L_n]_a\nonumber
&=&
\sum_{b=1}^n \ 
\sum_{(q,z)\in\CCP_{d-1}(v_n,v_b)} \sgn(q) \, \beta(q) \, \sgn(z) \, \beta(z) \, x_{b,a}\nonumber
\\
&\stackrel{(\ast)}{=}&
\sum_{b=1}^n \ 
\sum_{\substack{(q,z)\in\CCP_{d-1}(v_n,v_b)\\v_a\notin q, \ v_a\notin z}} \sgn(q) \, \beta(q) \, \sgn(z) \, \beta(z) \, x_{b,a}\nonumber
\\
&=&
- \sum_{(q,z')\in\CCP_d(v_n,v_a)} \sgn(q) \, \beta(q) \, \sgn(z') \, \beta(z')
\ = \ - [r_{n,d}]_a.\label{eq:rndrnd}
\end{eqnarray}
Since $q$ and $z$ in the sums are vertex-disjoint, to see $(\ast)$, it suffices to prove
\begin{eqnarray*}
&& \sum_{b=1}^n \ \sum_{\substack{(q,z)\in\CCP_{d-1}(v_n,v_b)\\v_a \in q}} \sgn(q) \, \beta(q) \, \sgn(z) \, \beta(z) \, x_{b,a}
\\
&=&
-\sum_{b'=1}^n \ \sum_{\substack{(q',z')\in\CCP_{d-1}(v_n,v_{b'})\\v_a \in z'}} \sgn(q') \, \beta(q') \, \sgn(z') \, \beta(z') \, x_{b',a}.
\end{eqnarray*}
This is proved analogously to the proof of Theorem~\ref{thm:CayleyHamilton}
via a weight-preserving sign-inverting bijection between the summands on the left-hand side and on the right-hand side.
Let $T_a = \{(q,z)\in\CCP_{d-1}(v_n,.)\mid v_a \in q\}$,
and let $S_a = \{(q',z')\in\CCP_{d-1}(v_n,.) \mid v_a \in z'\}$
where $\CCP_{d-1}(v_n,.)=\bigcup_{j\in\{1,\ldots,n\}}\CCP_{d-1}(v_n,v_j)$.
It will be important that $a\neq n$ by assumption.

For $(q,z)\in T_a$, let $v_b$ be the head of~$z$.
We traverse the edges of $z$ from $v_n$ to $v_b$ (these can be zero edges), then the edge $(v_b,v_a)$, and then the cycle $c$ in $q$ that contains $v_a$ (this can be a loop). This is a so-called directed tadpole graph, see Figure~\ref{fig:tadpole}.
\begin{figure}[htbp]
    \centering
    \begin{subfigure}[b]{0.45\textwidth}
        \centering
\begin{tikzpicture}[xscale=1.25]
\node[draw,circle,inner sep=0,minimum size=0.5cm] (1) at (1,0) {$v_1$};
\node[draw,circle,inner sep=0,minimum size=0.5cm] (2) at (2,0) {$v_2$};
\node[draw,circle,inner sep=0,minimum size=0.5cm] (3) at (3,0) {$v_3$};
\node[draw,circle,inner sep=0,minimum size=0.5cm] (4) at (4,0) {$v_4$};
\node[draw,circle,inner sep=0,minimum size=0.5cm] (5) at (5,0) {$v_5$};
\node[draw,circle,inner sep=0,minimum size=0.5cm] (6) at (6,0) {$v_6$};
\draw[-Triangle,dashed] (1) to[bend left] (2);
\draw[-Triangle,dashed] (2) to[bend left] (1);
\draw[-Triangle,dashed] (3) to[bend left] (4);
\draw[-Triangle,dashed] (4) to[bend left] (3);
\draw[-Triangle] (6) to (5);
\draw[-Triangle,big dots,draw opacity=0] (5) to (4); 
\end{tikzpicture}
        \caption{$n=6$, $a=4$, $b=5$, $b'=3$, $z=(v_6,v_5)$, $q=\{(\!(v_1,v_2)\!),(\!(v_3,v_4)\!)\}$. We have $c=(\!(v_3,v_4)\!)$.}
        \label{fig:tadpoleleft}
    \end{subfigure}
    \hfill
    \begin{subfigure}[b]{0.45\textwidth}
        \centering
\begin{tikzpicture}[xscale=1.25]
\node[draw,circle,inner sep=0,minimum size=0.5cm] (1) at (1,0) {$v_1$};
\node[draw,circle,inner sep=0,minimum size=0.5cm] (2) at (2,0) {$v_2$};
\node[draw,circle,inner sep=0,minimum size=0.5cm] (3) at (3,0) {$v_3$};
\node[draw,circle,inner sep=0,minimum size=0.5cm] (4) at (4,0) {$v_4$};
\node[draw,circle,inner sep=0,minimum size=0.5cm] (5) at (5,0) {$v_5$};
\node[draw,circle,inner sep=0,minimum size=0.5cm] (6) at (6,0) {$v_6$};
\draw[-Triangle,dashed] (1) to[bend left] (2);
\draw[-Triangle,dashed] (2) to[bend left] (1);
\draw[-Triangle,big dots,draw opacity=0] (3) to[bend left] (4);
\draw[-Triangle] (4) to[bend left] (3);
\draw[-Triangle] (6) to (5);
\draw[-Triangle] (5) to (4);
\end{tikzpicture}
        \caption{$n=6$, $a=4$, $b=5$, $b'=3$, $z'=(v_6,v_5,v_4,v_3)$, $q'=\{(\!(v_1,v_2)\!)\}$. We have $c=(\!(v_3,v_4)\!)$.}
        \label{fig:tadpoleright}
    \end{subfigure}
    \caption{The situation in the proof of Theorem~\ref{thm:wdetnsquare}.
The left and right subfigure are partner summands in the sign-inverting involution.}
    \label{fig:tadpole}
\end{figure}
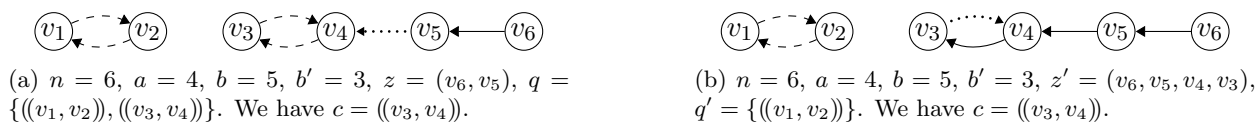
We construct $q'$ by removing the cycle~$c$ from~$q$. We construct $z'$ by starting with all edges in $z$, adding the edge $(v_b,v_a)$ to $z'$, and then traversing $c$ from $v_a$ until we reach the vertex $v_{b'}$ that satisfies $(v_{b'},v_a)\in c$, i.e., we add all but one of the edges of $c$ to~$z'$.
By construction, the set of edges of $z$ and $q$ together with $(v_b,v_a)$ equals the set of edges of $z'$ and $q'$ together with $(v_{b'},v_a)$.
Since $v_a \in z'$, we have $(q',z')\in S_a$.
The total number of edges stays the same, but the parity of the number of cycles changes in the process, which gives the desired sign change.

For $(q',z')\in S_a$, let $v_{b'}$ be the head of $z'$.
We have that $z'$ starts at $v_n$, $z'$ contains $v_a$ (note that $a\neq n$, hence $v_a$ has a well-defined predecessor in $z'$ that we call $v_b$), and ends at $v_{b'}$.
The edges of $z'$ from $v_a$ to $v_{b'}$ together with the edge $(v_{b'},v_a)$ form a cycle~$c$.
We remove all these edges from $z'$, and we also remove the edge $(v_b,v_a)$, to obtain the path~$z$ with head $v_b$.
We add the cycle $c$ to $q'$ to obtain~$q$.
By construction, the set of edges of $z'$ and $q'$ together with $(v_{b'},v_a)$ equals the set of edges of $z$ and $q$ together with $(v_b,v_a)$.
The total number of edges stays the same, but the parity of the number of cycles changes in the process, which gives the desired sign change.
This proves $(\ast)$, and hence \eqref{eq:rndrnd}, which is the left side of the block matrix in \eqref{eq:rnd}.

We now prove the right side of the block matrix in \eqref{eq:rnd}.
In the following, we group the partial cycle covers of length $d$ in $K_{n-1}$ into different sets according to the index $i$ of their highest numbered vertex:
\begin{equation}
\label{eq:computecpcnd}
[r_{n,d}]_n \stackrel{\eqref{eq:rndncpc}}{=} \cpc_{n-1,d} \ = \ \sum_{i=d}^{n-1}
\,\sum_{\substack{q \in \cyco_d(K_i)\\v_i \in q}} \sgn(q) \, \beta(q)
\stackrel{\eqref{eq:ridmCi}}{=} \sum_{i=d}^{n-1} r_{i,d-1} C_i.
\end{equation}

This proves the right side of the block matrix in \eqref{eq:rnd}, so the proof of \eqref{eq:rnd} is now complete.

Rewriting \eqref{eq:computecpcnd}, we have
\begin{equation}
\label{eq:computecpcndrewrite}
\cpc_{n,d} = \sum_{i=d}^n r_{i,d-1} C_i.
\end{equation}
We start by constructing an ABP that uses this identity and \eqref{eq:rnd} to compute $\cpc_{n,d}$.
To construct an ABP that computes all $r_{i,{d-1}}$, $d \leq i\leq n$,
we construct an ABP that computes all vectors $r_{i,j}$, $j < i \leq n$, $1 \leq j < d$ via \eqref{eq:rnd}.
By \eqref{eq:rndncpc} this computes all $\cpc_{i',j}$, $j \leq i' < n$, $1 \leq j < d$ along the way.
At degree $j$ we compute each vector $r_{i,j}$, $i \in \{j+1,\ldots,n\}$.
Each such vector has $i$ many elements, hence the total number of vertices in vertex layer $j$ is
$(j+1)+(j+2)+\cdots+n = \frac{(n-j)(n+j+1)}{2}$.
This is maximal in vertex layer 1, where we have $\binom{n+1}{2}-1$ many vertices.
Hence, $\w(\cpc_{n,d})\leq \binom{n+1}{2}-1$.
The total number of vertices (excluding source and sink) is $\sum_{j=1}^{d-1}\frac{(n-j)(n+j+1)}{2} = (d-1)\binom{n+1}{2}-\binom{d+1}{3}$.

We modify this ABP to compute the $\cpc_{n,j}$, $1 < j < d$, with one additional vertex each via \eqref{eq:computecpcndrewrite}.
Moreover, we use \eqref{eq:computecpcndrewrite} to also compute all $\cpc_{i,d}$, $d\leq i<n$, with one additional vertex each.
All traces $\cpc_{i,1}$, $1\leq i < n$, are computed along the way via \eqref{eq:rndncpc}.
The trace $\cpc_{n,1}$ can be computed with one additional vertex.
\end{proof}

For the sake of completeness, we now give the transition matrices for the construction in Theorem~\ref{thm:wdetnsquare}.
We write \eqref{eq:rnd} via block matrices:
\[
\big(r_{d,d-1}, r_{d+1,d-1}, r_{d+2,d-1}, \ldots, r_{n,d-1}\big)
\overbrace{\begin{pmatrix}
(0|C_d)      & (0|C_d)     &       \cdots     &  (0|C_d)     & (0|C_d)      & (0|C_d)     \\
(-L_{d+1}|0) & (0|C_{d+1}) &       \cdots     &  (0|C_{d+1}) & (0|C_{d+1})  & (0|C_{d+1}) \\
   0     & \ddots  &       \ddots     &  \vdots  & \vdots   & \vdots  \\
   0     & \ddots  &       \ddots     &  (0|C_{n-3}) & \vdots   & \vdots  \\
 \vdots  &   0     &         0        & (-L_{n-2}|0) & (0|C_{n-2})  & (0|C_{n-2}) \\
 \vdots  & \vdots  &       \vdots     &     0    & (-L_{n-1}|0) & (0|C_{n-1}) \\
   0     &   0     &         0        &     0    &    0     & (-L_n|0)
\end{pmatrix}}^{=: M_{n,d}}
\]
$
\hfill = \big(r_{d+1,d}, r_{d+2,d}, r_{d+3,d}, \ldots, r_{n,d}\big).\\
$
It is sometimes interesting to have an ABP whose edges are only labeled by constants or scalar multiples of variables, see Proposition~\ref{pro:CKV} below.
Since these matrices have only variables or their negations (or zero) in each cell,
and $C_d$ (the adjacency matrix of the last edge layer) also has only variables in each cell,
every edge in the ABP (besides the first edge layer) is either labeled with a variable $x_{i,j}$ or with $-x_{i,j}$. In the first edge layer, the edges are labeled with entries from the vectors $r_{n,1}$ for $2 \leq n \leq d$.
We have
\[
r_{n,1} \ = \ \big( -x_{n,1} \ \ -x_{n,2} \ \ \cdots \ \ -x_{n,n-1} \ \ \tr_{n-1} \big),
\]
where $\tr_n:=x_{1,1}+\cdots+x_{n,n}$.
The final matrix product is
\begin{equation}
\label{eq:numberofnonzeros}
(r_{2,1},r_{3,1},\ldots,r_{d,1})M_{d,2}M_{d,3}\cdots M_{d,d-1}C_d = \det_d.
\end{equation}
There is no need in the ABP to have edges labeled $\tr_1$, $\tr_2$, $\ldots$,
but we can instead label these edges with just $x_{1,1}$, $x_{2,2}$, $\ldots$,
and add edges with label 1 from each vertex that computes $x_{i,i}$ to the vertex that computes $x_{i+1,i+1}$,
exactly as in Figure~\ref{fig:charpolycharzero}.

The elementary symmetric polynomial $e_{n,d}$ is the restriction of $\cpc_{n,d}$ to diagonal matrices, and we have $\w(e_{n,d}) \leq n$.
Theorem~\ref{thm:wdetnsquare} gives $\w(\cpc_{n,d})\in O(n^2)$ via rather sparse matrices.
Moreover, the width of the construction is concentrated on the first layers.
Therefore, the following question seems reasonable.
\begin{openproblem}
\label{problem:wcpc}
Is $\w(\cpc_{n,d}) \in O(n)$?
\end{openproblem}

\section{Appendix}
\label{sec:appendix}
The arguments in this appendix are either well-known or variants of well-known results.
We included the appendix to make the paper more self-contained.

\medskip

In \S\ref{sec:intro} we defined a pABP as a $d$-tuple of $n \times n$ matrices of homogeneous linear polynomials
whose matrix product's bottom right entry is the computation result.
The width $\w(f)$ is the smallest $n$ so that $f$ is such a computation result.
Obviously, in the first matrix only the last row has any effect on the computation result,
and in the last matrix only the last column has any effect on the computation result,
so from now on we can assume that all other entries are zero.
These pABPs are in bijective correspondence to ABPs without constant edges in the obvious way:
The matrices specify the edge labels between the vertex layers.
In this section we interpret pABPs in this way, and not as a tuple of matrices.
This definition is for example used in \protect{\cite[Def 59]{KS14}}, under the name ABP.

\subsection{ABPs and pABPs give the same width}
Recall from \S\ref{sec:intro} the definition of an ABP~$G$,
which is a layered digraph labeled with homogeneous linear polynomials
with single source and sink that computes
\[
\sum_{p \in P(G,s,t)} \beta(p).
\]
Recall that the width of an ABP is the minimum of the number of vertices in all layers but the first and last.
The homogeneous width $\wh(f)$ is defined as the smallest $n$ such that $f$ is computed by a width $n$ ABP.

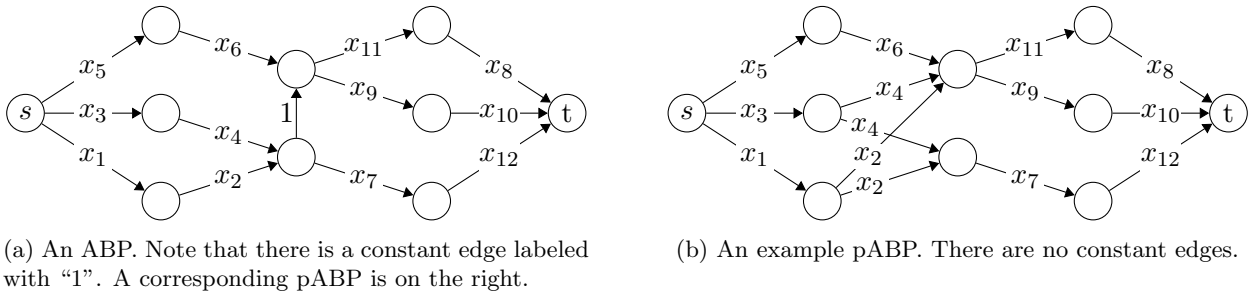
\begin{figure}
\centering
\begin{subfigure}[b]{0.47\textwidth}
\centering
\begin{tikzpicture}[xscale=1.8,yscale=1.17]
\node[draw,circle,inner sep=0,minimum width=5mm] (s) at (0,0) {$s$};
\node[draw,circle,inner sep=0,minimum width=5mm] (v1a) at (1,-1) {};
\node[draw,circle,inner sep=0,minimum width=5mm] (v1b) at (1,0) {};
\node[draw,circle,inner sep=0,minimum width=5mm] (v1c) at (1,1) {};
\node[draw,circle,inner sep=0,minimum width=5mm] (v2ab) at (2,-0.5) {};
\node[draw,circle,inner sep=0,minimum width=5mm] (v2bc) at (2,0.5) {};
\node[draw,circle,inner sep=0,minimum width=5mm] (v3a) at (3,-1) {};
\node[draw,circle,inner sep=0,minimum width=5mm] (v3b) at (3,0) {};
\node[draw,circle,inner sep=0,minimum width=5mm] (v3c) at (3,1) {};
\node[draw,circle,inner sep=0,minimum width=5mm] (t) at (4,0) {t};
\draw[-Triangle] (s) to node[pos=0.5,fill=white,inner sep=0,minimum size=0.5cm] {$x_1$} (v1a);
\draw[-Triangle] (s) to node[pos=0.5,fill=white,inner sep=0,minimum size=0.5cm] {$x_3$} (v1b);
\draw[-Triangle] (s) to node[pos=0.5,fill=white,inner sep=0,minimum size=0.5cm] {$x_5$} (v1c);
\draw[-Triangle] (v1a) to node[pos=0.5,fill=white,inner sep=0,minimum size=0.5cm] {$x_2$} (v2ab);
\draw[-Triangle] (v1b) to node[pos=0.5,fill=white,inner sep=0,minimum size=0.5cm] {$x_4$} (v2ab);
\draw[-Triangle] (v1c) to node[pos=0.5,fill=white,inner sep=0,minimum size=0.5cm] {$x_6$} (v2bc);
\draw[-Triangle] (v2ab) to node[left,pos=0.5,inner sep=0,minimum size=0.5cm] {$\ \ 1$} (v2bc);
\draw[-Triangle] (v2ab) to node[pos=0.5,fill=white,inner sep=0,minimum size=0.5cm] {$x_7$} (v3a);
\draw[-Triangle] (v2bc) to node[pos=0.5,fill=white,inner sep=0,minimum size=0.5cm] {$x_9$} (v3b);
\draw[-Triangle] (v2bc) to node[pos=0.5,fill=white,inner sep=0,minimum size=0.5cm] {$x_{11}$} (v3c);
\draw[-Triangle] (v3a) to node[pos=0.5,fill=white,inner sep=0,minimum size=0.5cm] {$x_{12}$} (t);
\draw[-Triangle] (v3b) to node[pos=0.5,fill=white,inner sep=0,minimum size=0.5cm] {$x_{10}$} (t);
\draw[-Triangle] (v3c) to node[pos=0.5,fill=white,inner sep=0,minimum size=0.5cm] {$x_{8}$} (t);
\end{tikzpicture}
\caption{An ABP. Note that there is a constant edge labeled with ``1''. A corresponding pABP is on the right.}
\label{subfig:compressed}
\end{subfigure}
\hfill
\begin{subfigure}[b]{0.47\textwidth}
\centering
\begin{tikzpicture}[xscale=1.8,yscale=1.17]
\node[draw,circle,inner sep=0,minimum width=5mm] (s) at (0,0) {$s$};
\node[draw,circle,inner sep=0,minimum width=5mm] (v1a) at (1,-1) {};
\node[draw,circle,inner sep=0,minimum width=5mm] (v1b) at (1,0) {};
\node[draw,circle,inner sep=0,minimum width=5mm] (v1c) at (1,1) {};
\node[draw,circle,inner sep=0,minimum width=5mm] (v2ab) at (2,-0.5) {};
\node[draw,circle,inner sep=0,minimum width=5mm] (v2bc) at (2,0.5) {};
\node[draw,circle,inner sep=0,minimum width=5mm] (v3a) at (3,-1) {};
\node[draw,circle,inner sep=0,minimum width=5mm] (v3b) at (3,0) {};
\node[draw,circle,inner sep=0,minimum width=5mm] (v3c) at (3,1) {};
\node[draw,circle,inner sep=0,minimum width=5mm] (t) at (4,0) {t};
\draw[-Triangle] (s) to node[pos=0.5,fill=white,inner sep=0,minimum size=0.5cm] {$x_1$} (v1a);
\draw[-Triangle] (s) to node[pos=0.5,fill=white,inner sep=0,minimum size=0.5cm] {$x_3$} (v1b);
\draw[-Triangle] (s) to node[pos=0.5,fill=white,inner sep=0,minimum size=0.5cm] {$x_5$} (v1c);
\draw[-Triangle] (v1a) to node[pos=0.3,fill=white,inner sep=0.5mm] {$x_2$} (v2ab);
\draw[-Triangle] (v1b) to node[pos=0.25,fill=white,inner sep=0.5mm] {$x_4$} (v2ab);
\draw[-Triangle] (v1a) to node[pos=0.3,fill=white,inner sep=0.5mm] {$x_2$} (v2bc);
\draw[-Triangle] (v1b) to node[pos=0.5,fill=white,inner sep=0,minimum size=0.5cm] {$x_4$} (v2bc);
\draw[-Triangle] (v1c) to node[pos=0.5,fill=white,inner sep=0,minimum size=0.5cm] {$x_6$} (v2bc);
\draw[-Triangle] (v2ab) to node[pos=0.5,fill=white,inner sep=0,minimum size=0.5cm] {$x_7$} (v3a);
\draw[-Triangle] (v2bc) to node[pos=0.5,fill=white,inner sep=0,minimum size=0.5cm] {$x_9$} (v3b);
\draw[-Triangle] (v2bc) to node[pos=0.5,fill=white,inner sep=0,minimum size=0.5cm] {$x_{11}$} (v3c);
\draw[-Triangle] (v3a) to node[pos=0.5,fill=white,inner sep=0,minimum size=0.5cm] {$x_{12}$} (t);
\draw[-Triangle] (v3b) to node[pos=0.5,fill=white,inner sep=0,minimum size=0.5cm] {$x_{10}$} (t);
\draw[-Triangle] (v3c) to node[pos=0.5,fill=white,inner sep=0,minimum size=0.5cm] {$x_{8}$} (t);
\end{tikzpicture}
\caption{An example pABP. There are no constant edges.\\\mbox{~}}
\label{subfig:uncompressed}
\end{subfigure}
\caption{An ABP and a pABP, both computing the same $f\in R[\mathbf x]_d$.
On the left: The ABP. On the right: The ABP after the constant edges are removed. The result is a pABP.}
\label{fig:compresseduncompressed}
\end{figure}

\begin{proposition}
\label{pro:wwh}
For every $f \in R[\mathbf x]_d$ we have $\w(f)=\wh(f)$.
\end{proposition}
\begin{proof}
Obviously, $\wh(f)\leq\w(f)$.
The other direction requires a classical construction.
Figure~\ref{fig:compresseduncompressed} gives an illustration.
If we have edges in an ABP \emph{within a vertex layer},
then these edges are labeled by elements from~$R$,
and the set of these edges is not allowed to contain a directed cycle.
These edges are called \emph{constant edges}.
Note that even though an ABP can have several vertices in vertex layer $0$ and vertex layer $d$,
there is only a single source and a single sink, and $\wh(f)$ is defined 
as the maximum number of vertices in a vertex layer $1 \leq k < d$,
i.e., we ignore vertex layer $0$ and $d$.
To an ABP we assign a pABP via the following process.
We choose a vertex $v$ that has no incoming constant edges (which must exist due to the absence of directed cycles), but at least one outgoing constant edge $(v,w)$ with label $\alpha\in R$.
We remove $(v,w)$, and for each edge $(u,v)$ with label $\beta(u,v)$ we add a new edge $(u,w)$ with label $\alpha\beta(u,v)$. If an edge $(u,w)$ already existed, then instead of having two parallel edges, we keep one edge whose label is the sum of the two labels.
If $v$ is the source vertex, then $v$ and $u$ are merged into one new source vertex, keeping all their outgoing edges,
and as before parallel edges are merged into single edges by adding their edge labels.
The whole process removes at least 1 constant edge, and does not change the set of vertices in any layer besides layer 0.
We continue this process until all constant edges are removed.
At the end, we delete all non-source vertices in vertex layer 0,
and all non-sink vertices in vertex layer $d$, so that we end up with a pABP.
Let $G_1,G_2,\ldots,G_j$ denote the list of ABPs arising between the steps of this process.
The invariant that is preserved between the steps is
\[
\sum_{p\in P(G_1,s,t)} \beta(p) \ = \ \sum_{p\in P(G_2,s,t)} \beta(p) \ = \ \cdots \ = \ \sum_{p\in P(G_j,s,t)} \beta(p),
\]
which proves the claim.
\end{proof}

\subsection{Degree and width}
\label{subsec:appendixduality}
\begin{proof}[Proof of \eqref{eq:duality}]
The statements for the degree are obvious, but
note that the second inequality can indeed be strict, because $R$ can have zero divisors.
The statements for the width are classical and can be proved by explicit constructions that appear for example in~\cite{Val79}.
The width-of-sum statement is obtained by taking a pABP $G_f$ for $f$ and a pABP $G_h$ for $h$ and get a pABP for the sum $f+h$ by identifying the two sources with each other and identifying the two sinks with each other.
The width-of-product statement is obtained by taking a pABP $G_f$ for $f$ and a pABP $G_h$ for $h$ and get a pABP for the product $fh$ by identifying the sink in $G_f$ with the source in $G_h$.
Both constructions are illustrated in Figure~\ref{fig:sumandproduct}.
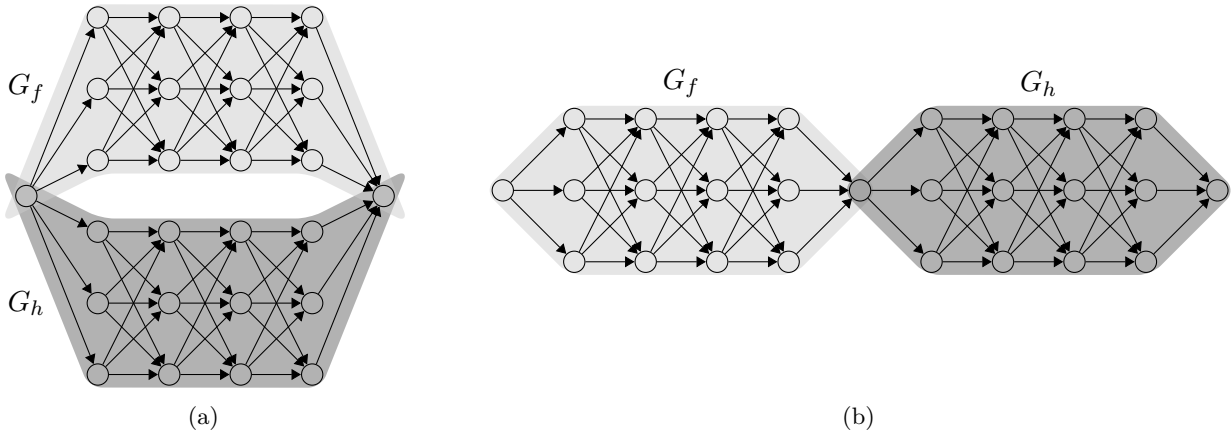
\begin{figure}[htbp]
    \centering
    \begin{subfigure}[b]{0.35\textwidth}
        \centering
\begin{tikzpicture}[scale=0.95]
\draw[draw=none,preaction={contour=-5pt,rounded corners=5,fill=black!60!white,opacity=0.5}] (0,1.5)--(1,-1)--(4,-1)--(5,1.5)--(4,1)--(1,1)--cycle;
\draw[draw=none,preaction={contour=-5pt,rounded corners=5,fill=black!20!white,opacity=0.5}] (0,1.5)--(1,2)--(4,2)--(5,1.5)--(4,4)--(1,4)--cycle;
\node[draw=none] at (0,3) {$G_f$};
\node[draw=none] at (0,0) {$G_h$};
\node[draw,circle,inner sep=0.1cm] (s) at (0,1.5) {};
\node[draw,circle,inner sep=0.1cm] (11) at (1,-1) {};
\node[draw,circle,inner sep=0.1cm] (12) at (1,0) {};
\node[draw,circle,inner sep=0.1cm] (13) at (1,1) {};
\node[draw,circle,inner sep=0.1cm] (21) at (2,-1) {};
\node[draw,circle,inner sep=0.1cm] (22) at (2,0) {};
\node[draw,circle,inner sep=0.1cm] (23) at (2,1) {};
\node[draw,circle,inner sep=0.1cm] (31) at (3,-1) {};
\node[draw,circle,inner sep=0.1cm] (32) at (3,0) {};
\node[draw,circle,inner sep=0.1cm] (33) at (3,1) {};
\node[draw,circle,inner sep=0.1cm] (41) at (4,-1) {};
\node[draw,circle,inner sep=0.1cm] (42) at (4,0) {};
\node[draw,circle,inner sep=0.1cm] (43) at (4,1) {};
\node[draw,circle,inner sep=0.1cm] (t) at (5,1.5) {};
\node[draw,circle,inner sep=0.1cm] (p11) at (1,2) {};
\node[draw,circle,inner sep=0.1cm] (p12) at (1,3) {};
\node[draw,circle,inner sep=0.1cm] (p13) at (1,4) {};
\node[draw,circle,inner sep=0.1cm] (p21) at (2,2) {};
\node[draw,circle,inner sep=0.1cm] (p22) at (2,3) {};
\node[draw,circle,inner sep=0.1cm] (p23) at (2,4) {};
\node[draw,circle,inner sep=0.1cm] (p31) at (3,2) {};
\node[draw,circle,inner sep=0.1cm] (p32) at (3,3) {};
\node[draw,circle,inner sep=0.1cm] (p33) at (3,4) {};
\node[draw,circle,inner sep=0.1cm] (p41) at (4,2) {};
\node[draw,circle,inner sep=0.1cm] (p42) at (4,3) {};
\node[draw,circle,inner sep=0.1cm] (p43) at (4,4) {};
\draw[-Triangle] (s) -- (11);
\draw[-Triangle] (s) -- (12);
\draw[-Triangle] (s) -- (13);
\draw[-Triangle] (11) -- (21);
\draw[-Triangle] (11) -- (22);
\draw[-Triangle] (11) -- (23);
\draw[-Triangle] (12) -- (21);
\draw[-Triangle] (12) -- (22);
\draw[-Triangle] (12) -- (23);
\draw[-Triangle] (13) -- (21);
\draw[-Triangle] (13) -- (22);
\draw[-Triangle] (13) -- (23);
\draw[-Triangle] (21) -- (31);
\draw[-Triangle] (21) -- (32);
\draw[-Triangle] (21) -- (33);
\draw[-Triangle] (22) -- (31);
\draw[-Triangle] (22) -- (32);
\draw[-Triangle] (22) -- (33);
\draw[-Triangle] (23) -- (31);
\draw[-Triangle] (23) -- (32);
\draw[-Triangle] (23) -- (33);
\draw[-Triangle] (31) -- (41);
\draw[-Triangle] (31) -- (42);
\draw[-Triangle] (31) -- (43);
\draw[-Triangle] (32) -- (41);
\draw[-Triangle] (32) -- (42);
\draw[-Triangle] (32) -- (43);
\draw[-Triangle] (33) -- (41);
\draw[-Triangle] (33) -- (42);
\draw[-Triangle] (33) -- (43);
\draw[-Triangle] (41) -- (t);
\draw[-Triangle] (42) -- (t);
\draw[-Triangle] (43) -- (t);
\draw[-Triangle] (s) -- (p11);
\draw[-Triangle] (s) -- (p12);
\draw[-Triangle] (s) -- (p13);
\draw[-Triangle] (p11) -- (p21);
\draw[-Triangle] (p11) -- (p22);
\draw[-Triangle] (p11) -- (p23);
\draw[-Triangle] (p12) -- (p21);
\draw[-Triangle] (p12) -- (p22);
\draw[-Triangle] (p12) -- (p23);
\draw[-Triangle] (p13) -- (p21);
\draw[-Triangle] (p13) -- (p22);
\draw[-Triangle] (p13) -- (p23);
\draw[-Triangle] (p21) -- (p31);
\draw[-Triangle] (p21) -- (p32);
\draw[-Triangle] (p21) -- (p33);
\draw[-Triangle] (p22) -- (p31);
\draw[-Triangle] (p22) -- (p32);
\draw[-Triangle] (p22) -- (p33);
\draw[-Triangle] (p23) -- (p31);
\draw[-Triangle] (p23) -- (p32);
\draw[-Triangle] (p23) -- (p33);
\draw[-Triangle] (p31) -- (p41);
\draw[-Triangle] (p31) -- (p42);
\draw[-Triangle] (p31) -- (p43);
\draw[-Triangle] (p32) -- (p41);
\draw[-Triangle] (p32) -- (p42);
\draw[-Triangle] (p32) -- (p43);
\draw[-Triangle] (p33) -- (p41);
\draw[-Triangle] (p33) -- (p42);
\draw[-Triangle] (p33) -- (p43);
\draw[-Triangle] (p41) -- (t);
\draw[-Triangle] (p42) -- (t);
\draw[-Triangle] (p43) -- (t);
\end{tikzpicture}
        \caption{}
    \end{subfigure}
    \hfill
    \begin{subfigure}[b]{0.6\textwidth}
        \centering
\raisebox{1.5cm}{\begin{tikzpicture}[scale=0.95]
\draw[draw=none,preaction={contour=-5pt,rounded corners=5,fill=black!20!white,opacity=0.5}] (0,0)--(1,-1)--(4,-1)--(5,0)--(4,1)--(1,1)--cycle;
\draw[draw=none,preaction={contour=-5pt,rounded corners=5,fill=black!60!white,opacity=0.5}] (5,0)--(6,-1)--(9,-1)--(10,0)--(9,1)--(6,1)--cycle;
\node[draw=none] at (2.5,1.5) {$G_f$};
\node[draw=none] at (7.5,1.5) {$G_h$};
\node[draw,circle,inner sep=0.1cm] (s) at (0,0) {};
\node[draw,circle,inner sep=0.1cm] (11) at (1,-1) {};
\node[draw,circle,inner sep=0.1cm] (12) at (1,0) {};
\node[draw,circle,inner sep=0.1cm] (13) at (1,1) {};
\node[draw,circle,inner sep=0.1cm] (21) at (2,-1) {};
\node[draw,circle,inner sep=0.1cm] (22) at (2,0) {};
\node[draw,circle,inner sep=0.1cm] (23) at (2,1) {};
\node[draw,circle,inner sep=0.1cm] (31) at (3,-1) {};
\node[draw,circle,inner sep=0.1cm] (32) at (3,0) {};
\node[draw,circle,inner sep=0.1cm] (33) at (3,1) {};
\node[draw,circle,inner sep=0.1cm] (41) at (4,-1) {};
\node[draw,circle,inner sep=0.1cm] (42) at (4,0) {};
\node[draw,circle,inner sep=0.1cm] (43) at (4,1) {};
\node[draw,circle,inner sep=0.1cm] (t) at (5,0) {};
\node[draw,circle,inner sep=0.1cm] (y11) at (6,-1) {};
\node[draw,circle,inner sep=0.1cm] (y12) at (6,0) {};
\node[draw,circle,inner sep=0.1cm] (y13) at (6,1) {};
\node[draw,circle,inner sep=0.1cm] (y21) at (7,-1) {};
\node[draw,circle,inner sep=0.1cm] (y22) at (7,0) {};
\node[draw,circle,inner sep=0.1cm] (y23) at (7,1) {};
\node[draw,circle,inner sep=0.1cm] (y31) at (8,-1) {};
\node[draw,circle,inner sep=0.1cm] (y32) at (8,0) {};
\node[draw,circle,inner sep=0.1cm] (y33) at (8,1) {};
\node[draw,circle,inner sep=0.1cm] (y41) at (9,-1) {};
\node[draw,circle,inner sep=0.1cm] (y42) at (9,0) {};
\node[draw,circle,inner sep=0.1cm] (y43) at (9,1) {};
\node[draw,circle,inner sep=0.1cm] (yt) at (10,0) {};
\draw[-Triangle] (s) -- (11);
\draw[-Triangle] (s) -- (12);
\draw[-Triangle] (s) -- (13);
\draw[-Triangle] (11) -- (21);
\draw[-Triangle] (11) -- (22);
\draw[-Triangle] (11) -- (23);
\draw[-Triangle] (12) -- (21);
\draw[-Triangle] (12) -- (22);
\draw[-Triangle] (12) -- (23);
\draw[-Triangle] (13) -- (21);
\draw[-Triangle] (13) -- (22);
\draw[-Triangle] (13) -- (23);
\draw[-Triangle] (21) -- (31);
\draw[-Triangle] (21) -- (32);
\draw[-Triangle] (21) -- (33);
\draw[-Triangle] (22) -- (31);
\draw[-Triangle] (22) -- (32);
\draw[-Triangle] (22) -- (33);
\draw[-Triangle] (23) -- (31);
\draw[-Triangle] (23) -- (32);
\draw[-Triangle] (23) -- (33);
\draw[-Triangle] (31) -- (41);
\draw[-Triangle] (31) -- (42);
\draw[-Triangle] (31) -- (43);
\draw[-Triangle] (32) -- (41);
\draw[-Triangle] (32) -- (42);
\draw[-Triangle] (32) -- (43);
\draw[-Triangle] (33) -- (41);
\draw[-Triangle] (33) -- (42);
\draw[-Triangle] (33) -- (43);
\draw[-Triangle] (41) -- (t);
\draw[-Triangle] (42) -- (t);
\draw[-Triangle] (43) -- (t);
\draw[-Triangle] (t) -- (y11);
\draw[-Triangle] (t) -- (y12);
\draw[-Triangle] (t) -- (y13);
\draw[-Triangle] (y11) -- (y21);
\draw[-Triangle] (y11) -- (y22);
\draw[-Triangle] (y11) -- (y23);
\draw[-Triangle] (y12) -- (y21);
\draw[-Triangle] (y12) -- (y22);
\draw[-Triangle] (y12) -- (y23);
\draw[-Triangle] (y13) -- (y21);
\draw[-Triangle] (y13) -- (y22);
\draw[-Triangle] (y13) -- (y23);
\draw[-Triangle] (y21) -- (y31);
\draw[-Triangle] (y21) -- (y32);
\draw[-Triangle] (y21) -- (y33);
\draw[-Triangle] (y22) -- (y31);
\draw[-Triangle] (y22) -- (y32);
\draw[-Triangle] (y22) -- (y33);
\draw[-Triangle] (y23) -- (y31);
\draw[-Triangle] (y23) -- (y32);
\draw[-Triangle] (y23) -- (y33);
\draw[-Triangle] (y31) -- (y41);
\draw[-Triangle] (y31) -- (y42);
\draw[-Triangle] (y31) -- (y43);
\draw[-Triangle] (y32) -- (y41);
\draw[-Triangle] (y32) -- (y42);
\draw[-Triangle] (y32) -- (y43);
\draw[-Triangle] (y33) -- (y41);
\draw[-Triangle] (y33) -- (y42);
\draw[-Triangle] (y33) -- (y43);
\draw[-Triangle] (y41) -- (yt);
\draw[-Triangle] (y42) -- (yt);
\draw[-Triangle] (y43) -- (yt);
\end{tikzpicture}}
        \caption{}
    \end{subfigure}
    \caption{The constructions in the proof of \eqref{eq:duality}. Addition is on the left, multiplication is on the right.}
    \label{fig:sumandproduct}
\end{figure}\end{proof}

\subsection{Width compared to affine ABPs}
An affine algebraic branching program (aABP) is a
directed acyclic graph with a single source $s$ and single sink~$t$.
The edges are labeled with (not necessarily homogeneous) polynomials of degree $\leq 1$.
The size of an aABP is the number of its vertices (not counting source or sink).
An aABP computes a polynomial via the sum of the products of the edge labels,
where the sum is over all paths from the source to the sink.
Note that it is not required that all $s$-$t$-paths have the same length.
We write $\asize(f)$ for the smallest size of an aABP computing $f$.
For $f\in R[\mathbf x]_d$, we clearly have
\begin{equation}
\label{eq:asizeleqw}
\asize(f)\leq d \cdot \w(f).
\end{equation}
The following proposition is classical and shows the other direction.
\begin{proposition}[Homogenization]
\label{pro:homogenization}
Let $h$ be a polynomial,
and let $f$ be a homogeneous component of $h$.
Then $\w(f) \leq \asize(h)$.
In particular, $\w(f) \leq \asize(f)$.
\end{proposition}
\begin{proof}
Let $d:=\deg(f)$.
We take a minimal size aABP $G$ that computes $h$
and we replace every vertex $v$ with $d+1$ many vertices $v_0,\ldots,v_d$.
We will set the edges so that if a vertex $v$ computes some polynomial $F$, then $v_i$ will compute the $i$-th homogeneous component of $F$, for $0\leq i \leq d$,
and we will not compute higher degree components.
If $\ell + \alpha$ is an edge label in $G$ from $v$ to $w$
with $\ell$ homogeneous linear and $\alpha$ constant,
then we add edges with label $\ell$ from each $v_i$ to $w_{i+1}$,
and we add edges with label $\alpha$ from each $v_i$ to $w_i$.
For the source $s$ we delete all $s_i$ with $i>0$,
and for the sink $t$ we delete all $t_i$ with $i<d$,
so that we have a unique source and sink.
This finishes the construction.
The result is an ABP that has at most $\asize(h)$ many vertices
in each vertex layer but the first and last, in other words, the width of the ABP is at most $\asize(h)$.
Hence $\w(h)\leq\asize(h)$.
\end{proof}

\subsection{Width compared to determinantal complexity}

\begin{proposition}[\cite{Val79}]
\label{pro:dcleqw}
Let $f\in R[\mathbf x]_d$. We have $\dc(f)\leq d \cdot \w(f)$.
\end{proposition}
\begin{proof}
Let $G$ be a width $n$ ABP that computes $f$.
Then $G$ has at most $(d-1)n+2$ many vertices.
We identify the source and the sink,
and to every other vertex we add a loop with label 1.
Let $G'$ denote the resulting graph, and let $M$ be its adjacency matrix.
Note that $G'$ has $(d-1)n+1\leq dn$ many vertices.
Since $\det(M)$ is the sum over all signed cycle covers of $G'$ (see \eqref{eq:detascyccov}),
we observe that the set of cycle covers in $G'$ is in weight-preserving bijection to the set of $s$-$t$-paths in $G$,
where the bijection is just adding/removing all loops.
Moreover, all $s$-$t$-paths have length $d$, hence $\det(M)=(-1)^{d+1} f$.
If $d$ is even, then we swap two rows of $M$ to obtain $M'$ with $\det(M')=f$.
\end{proof}

Almost 45 years lie between Proposition~\ref{pro:dcleqw} and its recent counterpart Proposition~\ref{pro:CKV}.

\begin{proposition}[Width variant of \protect{\cite[Thm~4.1]{CKV24}}]
\label{pro:CKV}
Let $R$ be a field, and $f\in R[\mathbf x]_d$. Then $\w(f) \leq d^4 \cdot \dc(f)$.
\end{proposition}
\begin{proof}
Let $M$ be an $s\times s$ matrix whose entries are in $R[\mathbf x]_{\leq d}$
such that $\det(M)=f$.
Let $M = M_1 + M_0$,
where $M_0$ is a matrix of constants, and $M_1$ is a matrix of homogeneous linear polynomials.
The matrix $M_0$ does not have full rank, because $\det(M)$ has constant part zero.
Since $R$ is a field,
we find $g,h\in\GL_s$ with $\det(gh)=1$ such that $g M_0 h = \begin{pmatrix}
0 & 0
\\
0 & I
\end{pmatrix}$.
Define $A,B,C,D$ via
\[
g M_1 h \ = \ 
\begin{pmatrix}
A & B
\\
C & -D
\end{pmatrix},
\]
where $D$ and $I$ have the same size, so that
\[
g M h = \begin{pmatrix}
A & B
\\
C & I-D
\end{pmatrix}.
\]
We have $N := (I-D)^{-1}=\sum_{i\geq 0}D^i$ in the ring of formal power series.
We have the Schur complement factorization:
\[
\begin{pmatrix}
A & B
\\
C & I-D
\end{pmatrix}
=
\begin{pmatrix}
I&B N\\
0&I
\end{pmatrix}
\begin{pmatrix}
A-BNC &0\\
0&I-D
\end{pmatrix}
\begin{pmatrix}
I & 0\\
DC & I
\end{pmatrix}.
\]
Since the first and third matrix in the factorization have determinant 1, it follows that
\[
\det\begin{pmatrix}
A & B
\\
C & I-D
\end{pmatrix}
=
\det
\begin{pmatrix}
A-BNC &0\\
0&I-D
\end{pmatrix} = \det(A-BNC)\,\det(I-D).
\]
The lowest nonzero homogeneous part of this determinant is $\det(A-BNC)$.
Since $f$ is homogeneous of degree $d$, it follows that $A-BNC$ is a $d\times d$ matrix.
Moreover, since $f$ is homogeneous of degree~$d$,
$f$ equals the homogeneous degree $d$ part of
$\det(A-B(\sum_{i=0}^{d-2}D^i)C)$,
because for $i>d-2$ the monomials of the matrix $BD^i C$ are of degree $>d$.

An aABP for an entry of $W = A-B(\sum_{i=0}^{d-2}D^i)C$ has size at most $s(d-1)$.
We take the ABP for $\det_d$ from Theorem~\ref{thm:wdetnsquare}, which only has variables or their negations or constants as edge labels.
We replace each variable $x_{i,j}$ with the aABP for $W_{i,j}$,
and each negated variable $-x_{i,j}$ with the aABP for $-W_{i,j}$:
Remove the edge $(a,b)$ that has label $x_{i,j}$, then add a copy of the aABP for $W_{i,j}$
and identify the source with $a$, and the sink with~$b$.
The result is an aABP of size at most $O(sd^4)$.
Since $\w(f)\leq \asize(f)$ due to Proposition~\ref{pro:homogenization}, it follows $\w(f)\leq O(sd^4)$.
\end{proof}
No purely combinatorial proof of Proposition~\ref{pro:CKV} is known.
Such an algorithm could be helpful in resolving the following open problem.
\begin{openproblem}
Does there exist a function $\xi:\IN\to\IN$
such that for every commutative ring $R$ and every $f\in R[\mathbf x]_d$ we have $\w(f)\leq \xi(d) \cdot \dc(f)$\,?
\end{openproblem}

\subsection{Exponents and fixed-parameter tractability}
\label{subsec:exponents}

Recall that $\VBFPT$ consists of those $f_{n,d}$ for which the set $\{\gamma(f_{.,d})\mid d \in \IN\}$ is bounded.
This nice phrasing is possible, because the model of computation is nonuniform, see \cite[footnote~5]{BE19}.
The following proposition shows the connection to the classical phrasing.

\begin{proposition}
\label{pro:VBFPT}
$\{\gamma(f_{.,d})\mid d \in \IN\}$ is bounded
 \ iff \ $\exists \xi:\IN\to\IN,c\in\IN \ \forall n,d: \w(f_{n,d}) \leq \xi(d) \cdot n^c$.
\end{proposition}
\begin{proof}
Given $\xi,c$ with $\forall n,d: \w(f_{n,d}) \leq \xi(d) \cdot n^c$.
Then $\forall n,d:\log_n(\w(f_{n,d})) \leq c + \log_n(\xi(d))$.
For all $d$, we have $\lim_{n\to\infty}\big(c + \log_n(\xi(d))\big) = c$, and hence $\forall d:\gamma(f_{.,d})\leq c$.

Let $c'\geq 0$ such that $\forall d:\gamma(f_{.,d})\leq c'$.
Hence, $\forall d,\varepsilon>0 \, \exists n_{d,\varepsilon} \, \forall n> n_{d,\varepsilon} : \log_n(\w(f_{n,d})) \leq c'+\varepsilon$.
In particular,
$\forall d \, \exists n_d \, \forall n> n_d : \log_n(\w(f_{n,d})) \leq c'+1$.
In other words, $\forall d \, \exists n_d \, \forall n> n_d : \w(f_{n,d})) \leq n^{c'+1}$.
Now, let $\xi(d) := \max\{\w(f_{n,d})) \mid n\leq n_d\}$.
Then
$\forall d : \w(f_{n,d})) \leq \xi(d) \cdot n^{c'+1}$.
The claim follows by setting $c=c'+1$.
Note that the argument works for any $c>c'\geq 1$, which gives different~$\xi$.
\end{proof}

\subsection{Algebraically closed fields}
\label{subsec:algclosed}
In this paper, $R$ is an arbitrary commutative ring.
Not every such ring can be embedded into a field.
For the sake of completeness, in this section we assume that $R$ is an algebraically closed field, and we show how some arguments can be made via Zariski density arguments.

\paragraph{The Cayley--Hamilton Theorem}
The Cayley--Hamilton theorem over algebraically closed fields can be proved as follows.
Note that for an $n\times n$ matrix $A$,
every eigenvalue $\la$ to an eigenvector $v$ satisfies $Av=\la v$,
in other words $(A-\la I_n)v = 0$, hence
$\la$ is a zero of $\det(A-t I_n) = 0$.
Conversely, if $\la$ is a zero of $\det(A-t I_n)$, then
$\det(A-\la I_n)=0$ and hence $A-\la I_n$ is not invertible,
which implies that there exists $v$ with $(A-\la I_n)v=0$,
hence $\la$ is an eigenvalue of $A$.

The set of matrices that satisfy the Cayley--Hamilton theorem is Zariski closed, because
the coefficients of the characteristic polynomial of $A$ depend polynomially on the entries of~$A$.
The set of diagonalizable matrices with distinct eigenvalues is Zariski dense in the space of all $n\times n$ matrices.
Hence, it suffices to prove the Cayley--Hamilton theorem for such matrices.
For such matrices, the characteristic polynomial of $A$ factors:
$\chi_A := \det(A-t I_n) = (t-\la_1)\cdots(t-\la_n)$,
where $\la_i$ are the eigenvalues of~$A$.
Let $v_i$ be eigenvectors of $\la_i$.
Now, for $P = \big(v_1 | \cdots | v_n\big)$, we have $A = P \, \diag(\la_1,\ldots,\la_n) \, P^{-1}$.
Hence, for $p(t)=t^k$ we have $p(A) = P \, \diag(\la_1^k,\ldots,\la_n^k) \, P^{-1}$,
and by linearity it holds for all polynomials $p(t)$ that
$p(A) = P \, \diag(p(\la_1),\ldots,p(\la_n)) \, P^{-1}$.
Since the $\la_i$ are the roots of $\chi_A$, it follows that
$\chi_A(A) = P \, \diag(0,\ldots,0) \, P^{-1} = 0$,
which finishes the proof of the Cayley--Hamilton theorem.
Since $R$ is infinite, by multivariate interpolation the Cayley--Hamilton theorem is also true as an identity of matrices of polynomials.

\paragraph{The matrix Girard--Newton Identities}
We now deduce the matrix Girard--Newton identities via a similar density argument from the classical Girard--Newton identities over algebraically closed fields.
Let $\trp_{n,j} := \tr(X_n^j)$, where $X_n = \big(x_{i,j}\big)$ is the $n \times n$ matrix of variables.
Our task is to prove 
\begin{equation}\label{eq:girardnewtonappendix}
\sum_{j=0}^{d-1} (-1)^{j} \cdot \cpc_{n,j} \cdot \trp_{n,d-j} \ \ + \ \ (-1)^{d} \cdot d \cdot \cpc_{n,d} \ = 0
\end{equation}
from the fact that
\begin{equation}\label{eq:girardnewtonsympolyappendix}
\sum_{j=0}^{d-1} (-1)^{j} \cdot e_j(\mathbf x) \cdot p_{d-j}(\mathbf x)
 \ \ + \ \ (-1)^{d} \cdot d \cdot e_d(\mathbf x) \ = 0,
\end{equation}
where $\mathbf x=(x_1,\ldots,x_n)$, and $e_j(\mathbf x) = \sum_{S\subseteq \{1,\ldots,n\}, |S|=j} e_{S_1}\cdots e_{S_j}$ is the elementary symmetric polynomial,
and $p_j(\mathbf x) = x_1^j + \cdots + x_n^j$ is the power sum polynomial.

Note that \eqref{eq:girardnewtonappendix} for diagonal matrices follows
directly from \eqref{eq:girardnewtonsympolyappendix} by observing
that for the diagonal matrix $D=\diag(x_1,\ldots,x_n)$ we have
$\cpc_{n,j}(D) = e_j(\mathbf x)$ and $\trp_{n,j}(D) = p_j(\mathbf x)$.
Since both $\cpc_{n,j}$ and $\trp_{n,j}$ are invariant under conjugation of their input variable matrices (note that $\cpc_{n,j}$ is a coefficient of $\det(X_n-t\cdot I_n)=\det(P X_n P^{-1}-t\cdot I_n)$),
and since \eqref{eq:girardnewtonappendix} is true when evaluated at diagonal matrices,
it follows \eqref{eq:girardnewtonappendix} holds when evaluated at any diagonalizable matrix.
Since $n\times n$ diagonalizable matrices are dense in the set of all $n\times n$ matrices,
\eqref{eq:girardnewtonappendix} holds for all $n\times n$ matrices.
Since $R$ is infinite, by multivariate interpolation \eqref{eq:girardnewtonappendix} is also true as an identity of polynomials.

\bigskip

\paragraph{Acknowledgments}
CI thanks
Robert Andrews for discussions about $\cpc_{n,d}$,
Markus Bl\"aser for discussions about $\VBFPT$,
and Darij Grinberg for feedback on a draft version, pointers to the literature, and making \cite{Gri25} available.
CI was partially supported by EPSRC grant EP/W014882/2.

\bigskip
\paragraph{Generative AI}
Generative AI erroneously claimed that Theorem~\ref{thm:CayleyHamilton} was already known and appeared in the Fadeev--LeVerrier algorithm.
During our investigation of this claim, we proved Corollary~\ref{cor:adjugate}.

{\footnotesize
\newcommand{\etalchar}[1]{$^{#1}$}

}

\end{document}